\documentclass[11pt]{article}
\usepackage[colorlinks=true,linkcolor=magenta,citecolor=magenta]{hyperref}
\usepackage[letterpaper, margin=1in]{geometry}

\title{Signal Recovery \\ with Non-Expansive Generative Network Priors}

\usepackage[utf8]{inputenc}
\usepackage[T1]{fontenc}
\usepackage{authblk}
\usepackage[numbers]{natbib}

\author{Jorio ~Cocola \\ Department of Mathematics, Northeastern University}

\usepackage{enumerate}
\usepackage{comment}
\usepackage{amsfonts,amsmath,bbm,amssymb, amsthm,amssymb, upgreek,mathtools}
\usepackage[linesnumbered,ruled,vlined]{algorithm2e}
\usepackage{float}
\usepackage{bm}

\usepackage[toc,page,header]{appendix}
\usepackage{minitoc}

\usepackage{amsfonts,amsmath,bbm,amssymb, amsthm,amssymb, upgreek,mathtools}
\usepackage{wrapfig}
\usepackage{comment}

\usepackage[english]{babel}
\usepackage{graphicx}
\usepackage{float, amssymb,amsmath}
\usepackage{ upgreek }
\usepackage{tcolorbox, color}
\usepackage{graphicx}
\usepackage{xcolor}
\usepackage{import}
\usepackage{paracol}

\usepackage{empheq}

\usepackage{thmtools}
\declaretheorem[
style=plain,
name=Theorem,
numberwithin=section
]{thm}
\declaretheorem[
style=plain,
name=Proposition,
numberlike=thm
]{prop}
\declaretheorem[
style=plain,
name=Lemma,
numberlike=thm
]{lemma}

\declaretheorem[
style=definition,
name=Definition,
numberlike=thm
]{mdef}

\newtheorem*{thm*}{Theorem}

\newtheorem{assumptionA}{Assumptions}

\newtheorem{assumptionC}{Assumption}

\newtheorem{assumptionB}{Assumptions}

\newtheorem{remark}{Remark}

\usepackage{varwidth}

\usepackage{tcolorbox}

\newcommand{\R}{\mathbb{R}}

\newcommand{\relu}{\mathsf{ReLU}}

\newcommand{\pa}{\partial}

\DeclareMathOperator{\EX}{\mathbb{E}}
\DeclareMathOperator{\PX}{\mathbb{P}}

\newcommand{\htilde}{\tilde{h}}

\newcommand{\eps}{\epsilon}

\newcommand{\Id}{\operatorname{I}}

\newcommand{\diag}{\text{diag}}

\newcommand{\xstar}{{x_{\star}}}

\newcommand{\ystar}{{{y_\star}}}

\newcommand{\Wji}{W_{j,i}}

\newcommand{\Wjpx}{W_{j,+,x}}
\newcommand{\Wjpy}{W_{j,+,y}}

\newcommand{\Wonepx}{W_{1,+,x}}

\newcommand{\Wpx}{W_{+,x}}

\newcommand{\Wpr}{W_{+,r}}
\newcommand{\Wps}{W_{+,s}}

\newcommand{\Lambdajx}{\Lambda_{j,x}}

\newcommand{\RRWDC}{R2WDC }
\newcommand{\RRWDCs}{R2WDC}

\newcommand{\fCS}{f_{\mathrm{cs}}}
\newcommand{\fPR}{f_{\mathrm{pr}}}
\newcommand{\fDen}{f_{\mathrm{den}}}
\newcommand{\fSpike}{f_{\mathrm{spiked}}}
\newcommand{\indx}{{\mathbbm{1}}}

\begin{document}

\doparttoc 
\faketableofcontents 

\maketitle

\begin{abstract}
    We study compressive sensing with a deep generative network prior. Initial theoretical 
    guarantees for efficient recovery from compressed linear measurements have been developed 
    for signals in the range of a ReLU network with Gaussian weights and logarithmic expansivity:
    that is when each layer is larger than the previous one by a logarithmic factor. It was later shown that
    constant expansivity is sufficient for recovery. It has remained open whether the expansivity 
    can be relaxed, allowing for networks with contractive layers (as often the case of real 
    generators). In this work we answer this question, proving that a signal in the range of a 
    Gaussian generative network can be recovered from few linear measurements provided that the 
    width of the layers is proportional to the input layer size (up to log factors). This condition 
    allows the generative network to have contractive layers. Our result is based on showing that 
    Gaussian matrices satisfy a matrix concentration inequality which we term \textit{Range Restricted 
    Weight Distribution Condition} (\RRWDCs) and that weakens the \textit{Weight Distribution Condition} (WDC) 
    upon which previous theoretical guarantees were based. The WDC has also been used to 
    analyze other signal recovery problems with generative network priors. By replacing the
    WDC with the \RRWDCs, we are able to extend previous results for signal recovery with expansive generative 
    network priors to non-expansive ones. We discuss these extensions for phase retrieval, 
    denoising, and spiked matrix recovery.
\end{abstract}

\section{Introduction}

The compressed sensing problem consists in estimating a signal $\ystar \in \R^{n}$ from (possibly) noisy linear measurements
\begin{equation*}
    b = A \ystar + \eta
\end{equation*}
where $A \in \R^{m \times n}$ is the measurements matrix, $m < n$ and $\eta \in \R^{m}$ is the noise.

To overcome the ill-posedness of the problem,
structural priors on the unknown signal $\ystar$ need to be enforced.
One now classical approach assumes that the target signal $\ystar$ is sparse with respect to a given basis.
In the last 20 years, efficient reconstruction algorithms have been developed that  provably estimate $s$-sparse signals
in $\R^n$ from $m = \mathcal{O}(s \log n)$ random measurements \citep{candes2006stable,donoho2006most}.

Another approach recently put forward, leverages trained generative networks. These networks are trained, in an unsupervised manner,
to generate samples from a target distribution of signals.
Assuming $\ystar$ belongs to the same distribution used to train a generative network $G: \R^k \to \R^n$ with $k \ll n$,
an estimate of $\ystar$ can be found by searching the input $\hat{x}$ (``latent code'') of $G$ that minimizes the reconstruction error
\begin{align}
    \tilde{x} &= \arg \min_{x \in \R^{x}} \fCS(x) := \frac{1}{2} \| b - A G(x) \|_2^2,  \label{eq:bora} \\
    \ystar &\approx G(\tilde{x}). \nonumber
\end{align}
As empirically demonstrated in \citep{bora2017compressed}, the minimization problem \eqref{eq:bora} can be solved efficiently
by gradient descent methods. Moreover,  solving \eqref{eq:bora} can effectively regularize the solution of the compressed sensing problem, significantly outperforming sparsity-based algorithms 
in the low measurements regime\cite{bora2017compressed}. Generative network based inversion algorithms have  been subsequently developed for
a variety of signal recovery problems, demonstrating their  potential to outperform inversion algorithms based on non-learned
(hand-crafted) priors \cite{oscar2018phase,mosser2020stochastic,menon2020pulse,heckel2021rate,pan2021do,mardani2018deep}.
For a recent overview see \cite{ongie2020deep}.




The optimization problem \eqref{eq:bora} is in general non-convex and gradient-based methods could get stuck in local minima.
To better understand the empirical success of \eqref{eq:bora},
in \cite{hand2019globalIEEE} the authors established theoretical guarantees for the noiseless compressed sensing problem ($\eta = 0$)
where $G: \R^k \to \R^n$ is a $d$-layer $\relu$ network of the form:
\begin{equation}\label{eq:Gx}
   G(x) = \relu(W_d \cdots \relu(W_2 \relu(W_1 x)))
\end{equation}
with $W_{i} \in \R^{n_i \times n_{i-1}}$, $n_0 = k$, $n_d = n$, and $\relu(z) = \max(z,0)$ is applied entrywise.
The authors of \cite{hand2019globalIEEE} used a probabilistic model for the generative network $G$ and measurement matrix $A$.
They assumed that each layer $W_i$ has independent Gaussian entries and is \textit{strictly expansive}. Specifically it holds that
\begin{equation}\label{eq:expansivityHand}
	n_i \geq n_{i-1} \cdot \log n_{i-1} \cdot \text{poly}(d)	 \qquad \text{for all} \; i = 1, \dots, d.
\end{equation}
Moreover, 
they considered $A$ to be a Gaussian matrix and $m \geq k \cdot \log n \cdot \text{poly}(d)$.
Under this probabilistic model it was shown in \cite{hand2019globalIEEE} that, despite its non-convexity,
$\fCS$ has a favorable optimization geometry and no spurious critical points exist apart from $\xstar$
and a negative multiple of it $- \rho_d \xstar$, where $\rho_d$ is a function of the depth $d$ of the network. 

The landscape analysis was later extended to recovery guarantees using a gradient based method in \citep{huang2021provably},
under the same probabilistic assumptions of \citep{hand2019globalIEEE}. In particular, \cite{huang2021provably}
has shown that there is an efficient gradient descent method (see Algorithm \ref{algo:subGrad}
in Section \ref{sec:MainRes}) that given as input $A, G$ and $b$ outputs a latent vector $\tilde{x}$
 such that $\|\ystar - G(\tilde{x}) \|_2 = O(\|\eta\|_2)$. This result demonstrated that efficient
 recovery is possible with a number of measurements which is information-theoretic optimal
 up to $log$-factors in $n$ and polynomials in $d$  ($m = \tilde{\Omega}(k)$).

 Generative networks used in practice though, have often contractive layers. Thus,
one major drawback of the theory developed in \cite{hand2019globalIEEE} is constituted
by the expansivity condition on the weight matrices \eqref{eq:expansivityHand}. Relaxing the condition \eqref{eq:expansivityHand}
and accommodating for generative networks with contractive layers
was formulated as an open problem in the survey paper\footnote{This open problem was also
proposed in the recent talk \cite{dimakis_talk}.} \cite{ongie2020deep}.

An initial positive result on this problem, came from \cite{daskalakis20}.
Using a refined analysis of the concentration of Lipschitz functions,
the authors proved that the results of \cite{hand2019globalIEEE,huang2021provably} hold true also for weight matrices satisfying
$ n_i \geq n_{i-1} \cdot \text{poly}(d).$
While not allowing for contractive layers, this condition removed the logarithmic expansivity requirement of \eqref{eq:expansivityHand}.

More recently, \cite{joshi2021plugincs} have studied the compressive sensing problem
with random generative network prior as in \cite{hand2019globalIEEE,huang2021provably}, and have
 shown that the expansivity condition can indeed be relaxed. They have provided an
 efficient iterative method that given as input $A, b$ and $G$, assuming that up to $log$-factors each layer width satisfies
  \begin{equation}\label{eq:expansivityJoshi}
 	n_i \gtrsim 5^i k,
 \end{equation}
 and the number of measurement satisfies
 \begin{equation}\label{eq:measurJoshi}
  m \gtrsim 2^d k,
 \end{equation}
  outputs a latent vector $\tilde{x}$ such that for $\ystar = G(\xstar)$
  it holds that $\|\ystar - G(\tilde{x}) \|_2 = O( 2^d \sqrt{\frac{k}{m}}  \|\eta\|_2)$
  with high probability\footnote{This algorithm and its analysis were initially given by the same
  authors for the denoising problem in \cite{joshi2021plugin}.}.
  Notice that the condition \eqref{eq:expansivityJoshi} while requiring
  the width to grow with the depth, can allow for contractive layers $n_i < n_{i-1}$.

\subsection{Our contributions}

It is natural to wonder whether the price to pay to remove the expansivity assumption
is indeed exponential in the depth $d$ of the network, as happens in the theoretical guarantees of \cite{joshi2021plugincs}.
In this paper, we answer this question.  Our main result is summarized below
and provides guarantees for solving compressed sensing with random generative network priors via a
gradient descent method (Algorithm \ref{algo:subGrad} in Section \ref{sec:MainRes}).

\begin{thm}[Informal version of Theorem \ref{thm:mainCS_rand}]\label{thm:informal}
Assume that $A$ has i.i.d. $\mathcal{N}(0,1/m)$ entries and each $W_{i}$ has i.i.d. $\mathcal{N}(0,1/{n_i})$
entries. Suppose that $\ystar = G(\xstar)$. Furthermore assume that, up to $log$-factors,
\begin{enumerate}
    \item $n_i \geq k \cdot \text{poly}(d)$;
    \item $m \geq k \cdot \text{poly}(d)$.
\end{enumerate}
Suppose that the noise error and the step size $\alpha > 0$ are small enough.
Then with high probability, Algorithm \ref{algo:subGrad} with input loss function $\fCS$,
step size $\alpha$ and number of iterations $T = $poly$(d)$, outputs an estimate $G(x_T)$
satisfying $\|G(x_T) -  \ystar \|_2 = O(\sqrt{\frac{k}{m}} \|\eta\|_2)$
\end{thm}

 Compared to \cite{huang2021provably} and \cite{daskalakis20}, our result do not
 require strictly expanding generative networks and allows for contractive layers.
 Furthermore, we show that the same algorithm proposed in \cite{huang2021provably} has a
 denoising effect, leading to a reconstruction of the target signal  $\ystar$ of the
 order $O(\sqrt{\frac{k}{m}} \|\eta\|_2)$ rather than only $O(\|\eta\|_2)$.

 Compared to the results of \cite{joshi2021plugincs} we show that it is sufficient for the
 width of the layers as well as the number of measurements to grow polynomially with the depth rather than exponentially.
 Similarly, compared to \cite{joshi2021plugincs}, we remove the exponential factor in the depth from the reconstruction error. 
 \smallskip


The analysis of \cite{hand2019globalIEEE} was based on a deterministic condition on the weight
matrices termed \textit{Weight Distribution Condition} (WDC). This condition, together with
a deterministic condition on $A$ (see Sec \ref{sec:detRec} for details),
was shown to be sufficient for the absence of spurious local minima in \eqref{eq:bora} and to be satisfied by expansive Gaussian random generative networks as \eqref{eq:Gx}. The WDC was also used in the subsequent
\cite{huang2021provably} to prove convergence of Algorithm \ref{algo:subGrad}.
Our main technical contribution is to show that the WDC can be replaced by a weaker
form of deterministic condition, termed \textit{Range Restricted Weight Distribution Condition} (\RRWDCs),
and still, obtain the absence of spurious local minima and
 recovery guarantees via Algorithm \ref{algo:subGrad}.
We will then show that random Gaussian networks satisfying the Assumption 1.\ of Theorem \ref{thm:informal} satisfy the \RRWDCs.

The framework introduced in \cite{hand2019globalIEEE}  was used in a number of recent works to
analyze other signal recovery problems with generative network priors, from one-bit recovery to
blind demodulation  \cite{qiu2020robust,ma2018invertibility,oscar2018phase,hand2019globalJoshi,song2019surfing,cocola2020nonasymptotic}.
These works considered expansive generative network priors, using the WDC and the results of \cite{hand2019globalIEEE} in their analysis.
Replacing the WDC with our \RRWDC we can extend the previous results in the literature to more
realistic (non-expansive) generative networks. This paper details these extensions
 for three representative signal recovery problems.

\begin{thm}\label{thm:informal2}
Suppose $G$ is random generative network as in \eqref{eq:Gx}, satisfying Assumption 1.\ of Theorem \ref{thm:informal}.
Then Algorithm \ref{algo:subGrad} with appropriate loss functions, step sizes, and number of steps,
succeed with high probability for Phase Retrieval, Denoising, and Spiked Matrix Recovery.\end{thm}

Our result on the denoising problem, implies a similar result on the inversion of a generative network.
The problem of inverting a generative neural network has important applications \cite{zhu2016generative,abdal2019image2stylegan,pan2021do}, and
 has been recently analyzed theoretically \cite{lei2019inverting, joshi2021plugin,aberdam2020and}.
 Our result shows that a random generative network can be efficiently inverted by gradient descent,
 even when containing contractive layers. This motivates the empirical use of gradient-based methods for inverting generative networks.


\subsection{Organization of the paper}

This paper is organized as follows. In Section \ref{sec:prelim} we introduce some notation used in the rest of the paper. In Section \ref{sec:MainRes} we formalize the compressed sensing problem with a generative network prior and describe an algorithm for the recovery. In Section \ref{sec:detRec} we describe our novel deterministic condition on the weights of the network (\RRWDCs) and provide theoretical guarantees for solving compressed sensing with a generative network prior satisfying this condition via the algorithm described in Section \ref{sec:MainRes}. Then in Section \ref{sec:randRec} we demonstrate that random non-expansive generative networks satisfy the \RRWDC with high probability. The appendix contains the full proof of the results described in the main text. Appendix \ref{appx:extensions} contains the extension of the theoretical guarantees for compressed sensing with a generative network prior to other signal recovery problems.

\section{Preliminaries}\label{sec:prelim}

We use $I_n$ to denote the $n \times n$ identity matrix.
For $j\geq 0$, we define the $j$-th sub-network $G_j : \R^{k} \to \R^{n_j}$ as $G_j(x) = \relu(W_j \cdots \relu(W_2 \relu(W_1 x)))$, 
with the convention that $G_0(x) = I_k x = x$. For a matrix $W \in \R^{n \times k}$, let $\diag(W x > 0)$
 be the diagonal matrix with $i$-th diagonal element equal to 1 if $(W x)_i > 0$ and 0 otherwise, 
 and $\Wpx = \diag(W x > 0) W$. We then define $\Wonepx = (W_1)_{+,x} = \diag(W_1 x > 0) W_1$ and 
\[
    \Wjpx = \diag(W_j W_{j-1,+,x} \cdots W_{2,+,x}W_{1,+,x})W_j.
\]
Finally, we let  $\Lambda_{0,x} = I_k$ and for $j\geq 1$  $\Lambdajx = \prod_{\ell = 1}^j W_{\ell,+,x}$ 
with $\Lambda_x = \Lambda_{d,x} = \prod_{\ell = 1}^d W_{\ell,+,x}$. 
Notice in particular that $G_j(x) = \Lambdajx x$ and $G(x) = \Lambda_x x$.

For $r, s$ nonzero vectors in $\R^{\ell}$, we define the matrix   
\begin{equation}\label{eq:defQ}
	    Q_{r,s} = \frac{\pi - \theta_{r,s}}{2 \pi} I_{\ell} + \frac{\sin \theta_{r,s}}{2 \pi} M_{\hat{r}\leftrightarrow \hat{s}}
	\end{equation}
	where $\theta_{r,s} = \angle (r,s)$, $\hat{r} = r/\|r\|_2$, $\hat{s} = s/\|s\|_2$, $I_{\ell}$ is the $\ell\times \ell$ 
	identity matrix and $M_{\hat{r}\leftrightarrow \hat{s}}$ is the matrix that sends $\hat{r} \mapsto \hat{s}$, $\hat{s} \mapsto \hat{r}$, 
	and with kernel span$(\{r,s\})^\perp$. If $r$ or $s$ are zero, then we let $Q_{r,s} = 0$.

\section{Problem statement and recovery algorithm}\label{sec:MainRes}
Consider a generative network $G: \R^{k} \to \R^{n}$ as in \eqref{eq:Gx}. 
The compressive sensing problem with a generative network prior can be formulated as follows. 
\begin{tcolorbox}
\begin{center}
\textbf{COMPRESSED SENSING WITH A DEEP GENERATIVE PRIOR}
\end{center}
 \begin{tabular}{rl}
 \textbf{Let}: & $G: \R^k \to \R^n$ generative network, $A \in \R^{m\times n}$ measurement matrix. \\
 \textbf{Let}: & $\ystar = G(\xstar)$ for some unknown $\xstar \in \R^k$.\\
 {} & {} \\
 \textbf{Given}: & $G$ and $A$. \\
  \textbf{Given}: & {Measurements} $b = A \ystar + \eta \in \R^{m}$ {with} $m \ll n$ and  $\eta \in \R^{m}$ {noise}.\\
   {} & {} \\
  \textbf{Estimate}: & $\ystar$.
  \end{tabular}
\end{tcolorbox}
\smallskip 


To solve the compressed sensing problem with deep generative prior $G$, in \cite{huang2021provably}, 
the authors propose the gradient descent method described in Algorithm \ref{algo:subGrad} with objective function $f = \fCS$. 
This algorithm attempts to minimize the objective function  $\fCS$ in \eqref{eq:bora}. 
Because of the $\relu$ activation function, the loss function $\fCS$ is nonsmooth. 
Algorithm \ref{algo:subGrad} therefore resorts to the notion of \textit{Clarke subdifferential}. Indeed, being continuous and piecewise smooth, 
at every point $x \in \R^k$, $\fCS$   admits a Clarke subdifferential given by\footnote{For details see for example \cite{clason2017nonsmooth}.}:
\begin{equation}\label{eq:subdiff}
	\pa \fCS(x) = \text{conv} \big\{ \lim_{p \to \infty} \nabla \fCS(x_p): \; x_p \to x, \, x_p \in \text{dom}(\nabla \fCS) \big\},
\end{equation}
where with $\text{conv}(\cdot)$ we denote the convex hull and with $\text{dom}(\nabla f)$ the subset of $\R^k$ (with full measure by Rademacher's theorem) where $f$ is differentiable.  The vectors $v_x \in \pa \fCS(x)$ are called the \textit{subgradients} 
of $\fCS$ at $x$, and at a point $x$ where $\fCS$ is differentiable it holds that $\pa \fCS(x) = \{\nabla \fCS(x) \}$.

\begin{algorithm}
\DontPrintSemicolon
\SetAlgoLined
\KwIn{Objective function $f$, initial point ${x}_0 \in \R^k \setminus \{0\}$ and step size $\alpha$}
\KwOut{An estimate of the target signal $\ystar = G(\xstar)$ and latent vector $\xstar$}
\For{$t = 0,1, \dots $}{
  \lIf{ $f(-x_t) < f(x_t)$}{
  $\tilde{x}_t \gets - x_t$}
  \lElse{ 
      $\tilde{x}_t \gets x_t$
    }
    Compute $v_{\tilde{x}_t} \in \pa f(\tilde{x}_t)$\\
    ${x}_{t+1} \gets \tilde{x}_t - \alpha v_{\tilde{x}_t}$
}
\Return{$x_t, G(x_t)$}\;
\caption{{\sc Subgradient Descent} \cite{huang2021provably}}
\label{algo:subGrad}
\end{algorithm}

Notice that, as described in line 5, Algorithm \ref{algo:subGrad} corresponds to a subgradient descent method with constant step size $\alpha$. Before taking a step in the direction of the subgradient though, the algorithm checks whether the objective function at the current state $x_t$ has a larger value than the value at its negative $- x_t$, and if so it updates the current state with its negative (line 3-4). This negation step allows the algorithm to escape the spurious critical point in a neighborhood of $- \rho_d \xstar$ where $\rho_d \in (0,1)$, and it is motivated by the landscape analysis of $\fCS$ under the deterministic and probabilistic assumptions that we describe in the coming sections.


\section{Recovery guarantees under deterministic conditions}\label{sec:detRec}

The strategy taken in \cite{hand2019globalIEEE} and \cite{huang2021provably} to analyze the landscape of the minimization problem \eqref{eq:bora} and the convergence of Algorithm \ref{algo:subGrad}, consists in identifying a set of deterministic conditions on the measurements matrix $A$ and the generative network $G$, that ensure that the objective function $\fCS$ is well behaved and Algorithm \ref{algo:subGrad} converges efficiently to an estimate of $\xstar$ and $\ystar$. These conditions are then shown to hold with high probability under probabilistic models for $A$ and $G$. 
This is akin to the results on compressed sensing with sparsity where, for example, recovery guarantees were developed under the Restricted Isometry Property \cite{candes2006robust}.

The first condition, introduced in \cite{hand2019globalIEEE}, is on the measurement matrix $A$ and ensures that $A^T A$
behaves like an isometry over differences of points in the range of a generative network $G$.
\begin{mdef}[RRIC \cite{hand2019globalIEEE}]
	A matrix $A \in \R^{m \times n}$ satisfies the \textbf{Range Restricted Isometry Condition} (RRIC) with respect to $G$ with constant $\eps$ if for all $x_1, x_2, x_3, x_4 \in \R^{k}$, it holds that 
	\begin{equation*}
     \Big| \langle \big( A^T A - I_n \big)\big( G(x_1) - G(x_2)\big),  G(x_3) - G(x_4) \rangle \Big| \leq \eps \|G(x_1) - G(x_2) \| \|G(x_3) - G(x_4) \|
	\end{equation*}
\end{mdef}

The second deterministic condition introduced in \citep{hand2019globalIEEE} is on the weight matrices of $G$, ensures that they are approximately distributed like a Gaussian, and allows the control of how the layers of the network distort angles.

\begin{mdef}[WDC \cite{hand2019globalIEEE}]\label{def:WDC}
	We say that a generative network $G$ as in \eqref{eq:Gx}, 
	satisfies the \textbf{Weight Distribution Condition }(WDC) with constant $\epsilon > 0$ 
	if for all $i = 1, \dots , d$, for all $r, s \in \R^{n_{i - 1}}$:
	\begin{equation}\label{eq:WDC}
	    \| (W_{i})_{+, r}^T (W_{i})_{+,s} - Q_{r,s}\|_2 \leq \eps,
	\end{equation}	
\end{mdef}
Strictly speaking, in \cite{hand2019globalIEEE} the authors define the WDC as a property of a single weight matrix $W$, 
and then assume that the WDC is satisfied at each layer $W_i$ of $G$. This is equivalent  
to the definition above and simplifies the introduction of a novel, weaker, condition on the weight matrices, the \RRWDC below. 

\begin{mdef}[\RRWDCs]
    We say that a generative network $G$ as in \eqref{eq:Gx}, satisfies the \textbf{Range Restricted Weight Distribution Condition }(\RRWDCs) 
    with constant $\eps > 0$ if for all $i = 1, \dots , d$, and for all $x,y, x_1, x_2, x_3, x_4 \in \R^{k}$ , it holds that
    \begin{equation}\label{eq:RRWDC}
    \begin{aligned}
    \big| \langle \big((W_i)_{+, r}^T &(W_i)_{+, s} - Q_{r,s}\big)u, v\rangle \big| \leq \epsilon \|u\| \|v\|,\\
        \text{where}\;\; & r = G_{i-1}(x), \\ & s = G_{i-1}(y),   \\
        & u = G_{i-1}(x_1) - G_{i-1}(x_2),  \\
        \text{and}\;\; & v = G_{i-1}(x_3) - G_{i-1}(x_4).
    \end{aligned}
    \end{equation}
    \end{mdef}

    Notice that the \RRWDC is weaker than the WDC. Indeed,  \eqref{eq:WDC} and \eqref{eq:RRWDC} 
    are equivalent for $i=1$, but for $i \geq 2$ equation \eqref{eq:WDC} requires $(W_{i})_{+,r}^t (W_{i})_{+,s}$ 
    to be close to the matrix $Q_{r,s}$ for any vector $r, s \in \R^{n_{i-1}}$ and when acting on any vector $u, v \in \R^{n_{i-1}}$,  
    while equation \eqref{eq:RRWDC} requires $(W_{i})_{+,r}^t (W_{i})_{+,s}$ to be close to the matrix $Q_{r,s}$ only for vectors $r, s$ 
    on the range of $G_{i-1}$ and when acting on vectors $u, v \in \R^{{n_{i-1}}}$ given by the difference of points on the range of $G_{i-1}$.     
    
   Our first technical result provides theoretical guarantees for efficiently estimating a target signal $\ystar$ on the range of a generative network from few linear measurements under the RRIC and the \RRWDC.

   \begin{thm}\label{thm:mainCS_det}
     Suppose $d \geq 2$, and $A$ and $G$ satisfy the RRIC and the \RRWDC with constant $\eps < K_1 / d^{90}$. Assume that $\|\eta \|_2 \leq \frac{K_2 \|\xstar\|_2}{d^{42} 2^{d/2}}$. Let  $\{x_t\}$ be the iterates generated by Algorithm \ref{algo:subGrad} 
     with loss function $\fCS$, initial point ${x}_0 \in \R^k \setminus \{0\}$ and step size $\alpha = K_3 \frac{2^d}{d^2}$. 
     Then there exists a number of steps $T$ satisfying $T \leq \frac{K_4 f(x_0) 2^d}{d^4 \eps \|\xstar\|_2^2}$ such that
     \[
        \|x_T - \xstar \|_2 \leq K_5 d^9 \|\xstar\|_2 \sqrt{\eps} + K_6 d^6 2^{d/2} \omega \|\eta\|_2.
     \]
     In addition, for all $t \geq T$, we have
     \begin{align*}
         \|x_{t+1} - \xstar \|_2 \leq C^{t + 1 - T} \|x_T - \xstar \|_2 + K_7 2^{d/2}  \| \eta \|_2, \\
         \|G(x_{t+1}) - \ystar\|_2 \leq \frac{1.2}{2^{d/2}} C^{t+1-T} \|x_T - \xstar\|_2 + 1.2 K_7  \|\eta \|_2,
          \end{align*}
          where $C = 1 - \frac{7}{8} \frac{\alpha}{2^d} \in (0,1)$. Here, $K_1, \dots, K_7$ are universal positive constants. 
    \end{thm}
    
    \medskip
    \begin{remark}
    The exponential factors $2^d$ appearing in the conditions and theses of the theorem are artifacts of the 
    scaling of the weights of the generative network. For example, the output $G(x)$ of the network  
    scales like $\|x\|_2/2^{d/2}$ (see for example Proposition \ref{prop:concGandTheta}). 
    Choosing the weights of the network to be $\{ \sqrt{2} \, W_i \}_{i \in [d]}$ would remove the $2^d$ factors in the above theorem. 
    \end{remark}
    
    This theorem shows that, despite the nonconvexity of the minimization problem \eqref{eq:bora}, if the RRIC and the \RRWDC hold with constant $\eps$, after $T = O(\eps^{-1})$ number of iterations the iterates of the subgradient descent method described in Algorithm \ref{algo:subGrad} enter in a region of local convergence around $\xstar$. Moreover, after a large enough number of steps, $G(x_t)$ gives an estimate of the target signal $\ystar$ up to the noise level $O(\|\eta \|)$. 
    
   \smallskip
  
  
 Theorem 3.1 in \cite{huang2021provably} shows that Theorem \ref{thm:mainCS_det} holds assuming that the RRIC and the WDC hold. Our first technical contribution is to show that the WDC in Theorem 3.1 of \cite{huang2021provably}, can be relaxed into the \RRWDCs. Relaxing the WDC into the \RRWDCs will enable the relaxing of the expansivity assumption needed to show that the WDC holds for Gaussian generative networks as we demonstrate in Section \ref{sec:randRec}.

We next describe the role of these deterministic conditions in the analysis of the problem \eqref{eq:bora}. The full proof of Theorem \ref{thm:mainCS_det} is given in Appendix \ref{appx:Props}.

\subsection{Global landscape analysis via the \RRWDC}

The analysis of \cite{hand2019globalIEEE} and \cite{huang2021provably} follows the approach recent line of works that analyze the global landscape geometry of non-convex optimization problems arising in statistical and signal recovery problems (see for example \cite{sun2016complete,sun2018geometric,ge2016matrix,ge2017no} and \cite{chi2019nonconvex} for an overview). The analysis roughly consists of two steps:
\begin{enumerate}[i)]
	\item Showing that $\fCS(x) \approx f_E(x)$ and $\pa \fCS(x) \approx h_{x}$ uniformly over $x$.
	\item Analyzing the global properties of $f_E(x)$ and $h_{x}$, and transfer them to $\fCS(x)$ and $h_{x}$ using the first step.
\end{enumerate}
Here $f_E(x)$ and $h_x$ are continuous functions of $x$, corresponding to the expected value of $\fCS(x)$ and $\pa \fCS(x)$ under Gaussian weights and measurement matrix $A$ (see next section for details) and zero noise. The RRIC and the WDC are used in \cite{hand2019globalIEEE} and \cite{huang2021provably} to obtain the uniform concentration in the first step, as well as directly proving convexity-like properties of $\pa \fCS(x)$ in the vicinity of $\xstar$.

To illustrate how the WDC and the \RRWDC come into play, consider for simplicity the noiseless case $\eta = 0$.
Then at a point $x \in \R^{k}$ where $G$ is differentiable, the gradient of $\fCS$ is given by 
\[
\begin{aligned}
	\nabla \fCS(x) &= \Lambda_{d,x}^T A^T (A \Lambda_{d,x} x - A \Lambda_{d,\xstar} \xstar), \\     
	&\approx 	\Lambda_{d,x}^T ( \Lambda_{d,x} x -  \Lambda_{d,\xstar} \xstar) 
\end{aligned} 
\]
where $\Lambda_{d,x}$ and $\Lambda_{d,\xstar}$ ar  defined in Section \ref{sec:prelim} and the approximation uses the fact that $A$ satisfies the RRIC with respect to $G$.
Then if $G$ satisfies the WDC we have that 
\[
\begin{aligned}
	&\nabla \fCS(x) \approx 	\Lambda_{d,x}^T ( \Lambda_{d,x} x -  \Lambda_{d,\xstar} \xstar) 
	\\ &= \Lambda_{d-1, x}^T (W_{d})_{+, G_{d-1}(x)}^T (W_{d}\big)_{+, G_{d-1}(x)} \Lambda_{d-1,x} x -  \Lambda_{d-1, x}^T (W_{d})_{+, G_{d-1}(x)}^T (W_{d}\big)_{+, G_{d-1}(\xstar)} \Lambda_{d-1,\xstar} \xstar \\
	&= \Lambda_{d-1, x}^T \Big[ Q_{G_{d-1}(x), G_{d-1}(x)} + O(\eps)\Big] \Lambda_{d-1,x} x -  \Lambda_{d-1, x}^T \Big[ Q_{G_{d-1}(x), G_{d-1}(\xstar)} + O(\eps)\Big]  \Lambda_{d-1,\xstar} \xstar
\end{aligned}
\]
where the last line used the WDC to control the concentration of 
$(W_{d})_{+, G_{d-1}(x)}^T (W_{d}\big)_{+, G_{d-1}(x)}$ and $(W_{d})_{+, G_{d-1}(x)}^T (W_{d}\big)_{+, G_{d-1}(\xstar)}$.  The resulting terms are then controlled again applying the WDC to the the other $d-1$ weights of $G$, so that proceeding by induction over $d$ one obtains
\begin{equation}\label{eq:approxh}
	\nabla \fCS(x) \approx h_x := \frac{1}{2^d} x - \frac{1}{2^d} \htilde_{x,\xstar}, 
\end{equation}
where $\htilde$ is a deterministic vector field defined in Appendix \ref{appx:Props}.

In Appendix \ref{appx:Props} we show that the \RRWDC can be used to control directly the concentration of the terms 
\[
	\Lambda_{d-1, x}^T (W_{d})_{+, G_{d-1}(x)}^T (W_{d}\big)_{+, G_{d-1}(x)} \Lambda_{d-1,x}x
\]
and 
\[
	\Lambda_{d-1, x}^T (W_{d})_{+, G_{d-1}(x)}^T (W_{d}\big)_{+, G_{d-1}(\xstar)} \Lambda_{d-1,\xstar} \xstar,
\]
around their expectation (with respect to $W_d$) obtaining in this way 
\[
\begin{aligned}
	&\nabla \fCS(x) \approx 	\Lambda_{d,x}^T ( \Lambda_{d,x} x -  \Lambda_{d,\xstar} \xstar) 
	\\ 
	&= \Lambda_{d-1, x}^T \big[ Q_{G_{d-1}(x), G_{d-1}(x)} \big] \Lambda_{d-1,x} x -  \Lambda_{d-1, x}^T \big[ Q_{G_{d-1}(x), G_{d-1}(\xstar)}\big]  \Lambda_{d-1,\xstar} \xstar \\ &\quad+ O(\eps \|\Lambda_{d-1, x}\| \|\Lambda_{d-1, x} x\|) + O(\eps \|\Lambda_{d-1, x}\| \|\Lambda_{d-1, \xstar} \xstar\|) 
\end{aligned}
\]
Then again applying the \RRWDC to the other layers of $G$, we can show that \eqref{eq:approxh} still holds. 
 We can then borrow the analysis of $h_{x}$ from \cite{huang2021provably} and obtain the same convergence guarantees. 
 
 The advantage of using the \RRWDC over the original WDC is that it is satisfied by random generative networks with contractive layers as we demonstrate in the next section.

\section{Recovery guarantees under probabilistic assumptions}\label{sec:randRec}

In this section we give probabilistic models for the measurement matrix $A$, generative network $G$, and noise vector $\eta$ that will ensure that the RRIC and the \RRWDC are satisfied with high probability and Algorithm \ref{algo:subGrad} efficiently estimate the target signal $\ystar$ up to an error of the order $\tilde{O}(\sqrt{{k}/{m}}\|\eta\|)$.

We make the following assumption on the sensing matrix $A \in \R^{m \times n}$.
\begin{assumptionA}\label{AssumptionA}
${}$

    \begingroup
     \renewcommand{\labelenumi}{\textbf{{A.\arabic{enumi}}}}
    \begin{enumerate}\label{AssumptionA}
        
        \item $A$ is independent from $\{W_{i}\}_{i=1}^d$.\label{assmp:A3}
        
        \item $A$ has i.i.d. $\mathcal{N}(0,1/m)$ entries.\label{assmp:A4}
        
        \item There are sufficient number of linear measurements:\label{assmp:A5}
        \begin{equation}\label{eq:measurements}
            m \geq \widehat{C}_{\eps} \cdot k \cdot \log \prod_{j=1}^{d} \frac{e \,n_i}{k},
        \end{equation}
        where $\widehat{C}_{\eps}$ depends polynomially on $\eps^{-1}$.
        \end{enumerate}
    \endgroup
    
    \end{assumptionA}

Under Assumptions \ref{AssumptionA}, the measurement matrix satisfies the RRIC with respect to $G$ with high probability.  

\begin{lemma}[Consequence of Proposition 6 in \cite{hand2019globalIEEE}]\label{lemma:RRIC_rand}
	Let Assumptions \ref{AssumptionA} be satisfied. Then $A$ satisfies the RRIC with constat $\eps > 0$ with respect to $G$, with probability at least
	\[
		1 - \hat{\gamma} e^{- \hat{c} \eps m}
	\]
	where $\hat{\gamma}$ and $\hat{c}$ are positive universal constants. 
\end{lemma}    
\begin{proof}
	This result is proved in Proposition 6 in \cite{hand2019globalIEEE} for a number of measurements $m$ satisfying $m \geq C_\eps' \cdot k \cdot d \cdot \log \prod_{j=1}^d n_j $ where $C_\eps'$ depends polynomially on $\eps$. To imporove the lower bound on $m$ to \eqref{eq:measurements} it is enough to follow the proof of Proposition 6 in \cite{hand2019globalIEEE} and use the sharper upper bound on the number of affine subspaces in the range of a gnerative network given in Lemma \ref{lemma:joshiRange}.
\end{proof}

We then provide a probabilistic model for a generative network $G: \R^k \to \R^{n}$ as in \eqref{eq:Gx}.

\begin{assumptionB}\label{AssumptionB}
${}$

    \begingroup
     \renewcommand{\labelenumi}{\textbf{{B.\arabic{enumi}}}}
    \begin{enumerate}\label{AssumptionB}
      
    \item Each weight matrix $W_i \in \R^{n_i \times n_{i-1}}$ have i.i.d. $\mathcal{N}(0,1/n_i)$ entries. \label{assmp:A1}
       
        \item The first layer satisfies $n_1 \geq \widetilde{C}_\eps \cdot k $, and for any $i = 2, \dots, d$:\label{assmp:A2}
        \begin{equation}\label{eq:expansivity}
            n_i \geq \widetilde{C}_\eps \cdot k \cdot \log \prod_{j=1}^{i-1} 
             \frac{e \,n_j}{k},
        \end{equation}
        where $\tilde{C}_{\eps}$ depends polynomially on $\eps^{-1}$.
        
        \item The $\{W_{j}\}_{j=1}^{d}$ are independent. \label{assmp:A2p}
    \end{enumerate}
    \endgroup
    
    \end{assumptionB}
  
  
  Under Assumptions \ref{AssumptionB}, the generative network $G$ satisfies the \RRWDC.
  \begin{lemma}\label{prop:RandRRWDC}
	Fix $0 < \eps < 1$. Consider a $d$-layer $\relu$ network $G$ with weight matrices $\{W_i\}_{i=1}^d$. 
	Assume that the $\{W_i\}_{i=1}^d$ satisfy Assumptions \ref{AssumptionB}. 
	Then $G$ satisfies the \RRWDC with constant $\eps$  with probability at least
\[
  1 - \gamma \Big(\frac{e n_1}{k}\Big)^{2k}  e^{- {c}_\eps n_1 } - {\gamma} \sum_{i = 2}^d  \Big(\frac{e \, n_i}{k+1}\Big)^{4 k} e^{- {c}_\eps n_i/2} \] 
 where $c_\eps$ depends polynomially on $\eps^{-1}$ and ${\gamma}$ is a positive absolute constant.
\end{lemma}

\medskip

  We finally conclude with some assumptions on the noise vector $\eta \in \R^m$.
    
\begin{assumptionC}\label{AssumptionC}
${}$
    The noise vector $\eta$ is independent from $A$ and the weights $\{W_{i}\}_{i=1}^d$    
\end{assumptionC}
   
The next lemma is used to bound the perturbation of the objective function $\fCS$ 
and its gradient due to the presence of the noise term $\eta$. 
These bounds are then used to show that Algorithm 1 leads to a reconstruction of $\ystar$ of the order $O(\sqrt{k/m} \|\eta\|)$.

\begin{lemma}\label{lemma:noise1}
	Suppose $G: \R^k \to \R^n$ satisfies the \RRWDC with $\eps < 1/(16 \pi d^2)^2$ and $d \geq 2$. Let $A \in \R^{m \times n}$ 
	be a matrix with i.i.d.\ entries $\mathcal{N}(0,1/m)$ and $\eta \in \R^{m}$ satisfies Assumption \ref{AssumptionC}. Let
	\begin{align}\label{eq:omega}
        \omega &:= \frac{2}{2^{d/2}}  \sqrt{\frac{13}{12}} \sqrt{\frac{k}{m} \log\Big(5 \prod_{j=1}^{d} \frac{e \,n_i}{k} \Big) }.
    \end{align} 
     Then with probability at least
	\[
		1 - e^{- \frac{k}{2} \log(5 \prod_{i=1}^{d} \frac{e \,n_i}{k})}
	\]
	for every $x \in \R^k$ we have that
	\begin{equation}\label{eq:noise1}
		\langle x, \Lambda_x^T A^T \eta \rangle \leq \omega  \| \eta \| \|x\|,
	\end{equation}
	 if in addition $G$ is differentiable at $x$ we also have that
	 \begin{equation}\label{eq:noise2}
	 	\| \Lambda_x^T A^T \eta \| \leq \omega  \| \eta \|.
	 \end{equation}
\end{lemma}

     Given the previous assumptions, we are now ready to state the main result of this section.
      
    \begin{thm}\label{thm:mainCS_rand}
     Suppose $d \geq 2$, $\eps < K_1 / d^{90}$ and $\omega \|\eta \|_2 \leq \frac{K_2 \|\xstar\|_2}{d^{42} 2^{d/2}}$ where $\omega$ is defined in \eqref{eq:omega}. 
     
     Assume that $A$, $G$ and $\eta$ satisfy Assumptions \ref{AssumptionA}, \ref{AssumptionB} and \ref{AssumptionC}. Then  with probability at least
     \begin{equation}\label{eq:probab_thm}
        1 -  \gamma \Big(\frac{e \, n_1}{k}\Big)^{2 k} e^{- {c_\eps} n_1 } - {\gamma} \sum_{i = 2}^d  \Big(\frac{e \, n_i}{k+1}\Big)^{4 k} e^{- {c}_\eps n_i/2}  - \hat{\gamma} e^{- \hat{c} \eps m} - e^{- \frac{k}{2} \log(5 \prod_{i=1}^{d} \frac{e \,n_i}{k})},
     \end{equation}
     where $\gamma$, $\hat{\gamma}$ and $\hat{c}$ are positive universal constants, 
     the following holds. Let  $\{x_t\}$ be the iterates generated by Algorithm \ref{algo:subGrad} 
     with loss function $\fCS$, initial point ${x}_0 \in \R^k \setminus \{0\}$ and step size $\alpha = K_3 \frac{2^d}{d^2}$. 
     There exists a number of steps $T$ satisfying $T \leq \frac{K_4 f(x_0) 2^d}{d^4 \eps \|\xstar\|_2}$ such that
     \[
        \|x_T - \xstar \|_2 \leq K_5 d^9 \|\xstar\|_2 \sqrt{\eps} + K_6 d^6 2^{d/2} \omega \|\eta\|_2.
     \]
     In addition, for all $t \geq T$, we have
     \begin{align*}
         \|x_{t+1} - \xstar \|_2 \leq C^{t + 1 - T} \|x_T - \xstar \|_2 + K_7 2^{d/2} \omega \| \eta \|_2, \\
         \|G(x_{t+1}) - \ystar\|_2 \leq \frac{1.2}{2^{d/2}} C^{t+1-T} \|x_T - \xstar\|_2 + 1.2 K_7 \omega \|\eta \|_2,
          \end{align*}
          where $C = 1 - \frac{7}{8} \frac{\alpha}{2^d} \in (0,1)$. Here, $K_1, \dots, K_7$ are universal positive constants. 
    \end{thm}
    \begin{proof}
    	Combining Lemma \ref{lemma:RRIC_rand}, Lemma \ref{prop:RandRRWDC} and Theorem \ref{thm:mainCS_det} we obtain Theorem \ref{thm:mainCS_rand} with $\omega = 1$ and probability at least
    \[
        1 -  \gamma \Big(\frac{e \, n_1}{k}\Big)^{2 k} e^{- {c_\eps} n_1 } - {\gamma} \sum_{i = 2}^d  \Big(\frac{e \, n_i}{k+1}\Big)^{4 k} e^{- {c}_\eps n_i/2}  - \hat{\gamma} e^{- \hat{c} \eps m}.
     \]
   	Inspecting the proof of Theorem 3.1 in \cite{huang2021provably}, it is easy to see that if Lemma \ref{lemma:noise1} holds, then the conclusions of Theorem \ref{thm:mainCS_rand} hold with $\omega$ given by \eqref{eq:omega} and probability  at least \eqref{eq:probab_thm}.
    \end{proof}
    
    
    \medskip
    \begin{remark}
    As for Theorem \ref{thm:mainCS_det}, the exponential factors $2^d$ are artifacts of the 
    scaling of the weights of the network. Had the entries of $W_i$ been drawn from $\mathcal{N}(0, 2/n_i)$  the $2^d$ factors would not be present.    \end{remark}
    
    \begin{remark}
        Notice that $4 k \log( {e n}/{(k+1)} \big) \leq 4 k {\log(n)}/{\log(2)}$ for every $n \geq 2$. 
        Thus if for every $i = 1, \dots, d$, it holds that 
        \begin{equation}\label{eq:addAssmp}
        \frac{n_i}{\log(n_i)} \geq \frac{16 \cdot k \cdot c_\eps^{-1} }{\log(2)}
        \end{equation}
         the conclusions of the theorem hold with nontrivial probability bounds. 
         In Appendix \ref{appx:example} we provide an example of a generative network $G$ with \textit{contractive} 
         layers satisfying both \eqref{eq:expansivity} and \eqref{eq:addAssmp}.
     \end{remark}
    
    \smallskip
    
    Theorem \ref{thm:mainCS_rand} provides guarantees for the efficient recovery of a signal $\ystar$ in the range of a generative network $G$ from few noisy linear measurements, using a nonconvex (sub)gradient descent method. Notice that the intrinsic dimension of the signal $\ystar$ is $k$ (the dimension of the latent space) and the number of measurements required $m$ is proportional to $k$ and information-theoretically optimal up to $\log$ factors in the widths of the network and polynomials in the depth. Notice moreover, that up to these factors, the width $n_i$ of each layer of the network is also required to be linear in $k$. This is necessary to ensure that each subnetwork $G_i: \R^k \to \R^{n_i}$ is invertible, and it is weaker than the assumptions in the previous works that required $n_i$ to be linear in $n_{i-1}$ in order to ensure the invertibility of every single layer. We leave for future works the establishing of sharper lower bounds on the network widths and number of measurements. 
    
    Limitations of the current and previous works on theoretical guarantees for signal recovery with generative networks are the Gaussian assumption on the weights and the absence of biases. Important directions of future research are the inclusion of biases in the generative network and the departure from the Gaussian weights assumptions for more realistic probabilistic models.

    	

\subsection*{Acknowledgement}
{I would like to thank Paul Hand for comments on an
earlier version of this manuscript, and Babhru Joshi for helpful discussions.}

\bibliographystyle{plain}
\bibliography{main}

\begin{thebibliography}{10}

\bibitem{abdal2019image2stylegan}
Rameen Abdal, Yipeng Qin, and Peter Wonka.
\newblock Image2stylegan: How to embed images into the stylegan latent space?
\newblock In {\em Proceedings of the IEEE/CVF International Conference on Computer Vision}, pages 4432--4441, 2019.

\bibitem{aberdam2020and}
Aviad Aberdam, Dror Simon, and Michael Elad.
\newblock When and how can deep generative models be inverted?
\newblock {\em arXiv preprint arXiv:2006.15555}, 2020.

\bibitem{bora2017compressed}
Ashish Bora, Ajil Jalal, Eric Price, and Alexandros~G Dimakis.
\newblock Compressed sensing using generative models.
\newblock In {\em Proceedings of the 34th International Conference on Machine Learning-Volume 70}, pages 537--546. JMLR. org, 2017.

\bibitem{candes2006robust}
Emmanuel~J Cand{\`e}s, Justin Romberg, and Terence Tao.
\newblock Robust uncertainty principles: Exact signal reconstruction from highly incomplete frequency information.
\newblock {\em IEEE Transactions on information theory}, 52(2):489--509, 2006.

\bibitem{candes2006stable}
Emmanuel~J Candes, Justin~K Romberg, and Terence Tao.
\newblock Stable signal recovery from incomplete and inaccurate measurements.
\newblock {\em Communications on Pure and Applied Mathematics: A Journal Issued by the Courant Institute of Mathematical Sciences}, 59(8):1207--1223, 2006.

\bibitem{chi2019nonconvex}
Yuejie Chi, Yue~M Lu, and Yuxin Chen.
\newblock Nonconvex optimization meets low-rank matrix factorization: An overview.
\newblock {\em IEEE Transactions on Signal Processing}, 67(20):5239--5269, 2019.

\bibitem{clason2017nonsmooth}
Christian Clason.
\newblock Nonsmooth {A}nalysis and {O}ptimization.
\newblock {\em arXiv preprint arXiv:1708.04180}, 2017.

\bibitem{cocola2020nonasymptotic}
Jorio Cocola, Paul Hand, and Vlad Voroninski.
\newblock Nonasymptotic guarantees for spiked matrix recovery with generative priors.
\newblock {\em Advances in Neural Information Processing Systems}, 33:15185--15197, 2020.

\bibitem{cocola2021no}
Jorio Cocola, Paul Hand, and Vladislav Voroninski.
\newblock No statistical-computational gap in spiked matrix models with generative network priors.
\newblock {\em Entropy}, 23(1):115, 2021.

\bibitem{daskalakis20}
Constantinos Daskalakis, Dhruv Rohatgi, and Emmanouil Zampetakis.
\newblock Constant-expansion suffices for compressed sensing with generative priors.
\newblock In H.~Larochelle, M.~Ranzato, R.~Hadsell, M.~F. Balcan, and H.~Lin, editors, {\em Advances in Neural Information Processing Systems}, volume~33, pages 13917--13926. Curran Associates, Inc., 2020.

\bibitem{dimakis_talk}
Alex Dimakis.
\newblock Deep generative models and unsupervised methods for inverse problems, \url{https://youtu.be/OsrR9Ar1tVc?t=2069}, October 2021.
\newblock In \textit{Algorithmic Advances for Statistical Inference with Combinatorial Structure}.

\bibitem{donoho2006most}
David~L Donoho.
\newblock For most large underdetermined systems of linear equations the minimal $\ell_1$-norm solution is also the sparsest solution.
\newblock {\em Communications on Pure and Applied Mathematics: A Journal Issued by the Courant Institute of Mathematical Sciences}, 59(6):797--829, 2006.

\bibitem{ge2017no}
Rong Ge, Chi Jin, and Yi~Zheng.
\newblock No spurious local minima in nonconvex low rank problems: A unified geometric analysis.
\newblock In {\em International Conference on Machine Learning}, pages 1233--1242. PMLR, 2017.

\bibitem{ge2016matrix}
Rong Ge, Jason~D Lee, and Tengyu Ma.
\newblock Matrix completion has no spurious local minimum.
\newblock {\em Advances in neural information processing systems}, 29, 2016.

\bibitem{hand2019globalJoshi}
Paul Hand and Babhru Joshi.
\newblock Global guarantees for blind demodulation with generative priors.
\newblock In {\em Advances in Neural Information Processing Systems}, pages 11531--11541, 2019.

\bibitem{oscar2018phase}
Paul Hand, Oscar Leong, and Vlad Voroninski.
\newblock Phase retrieval under a generative prior.
\newblock In {\em Advances in Neural Information Processing Systems}, pages 9136--9146, 2018.

\bibitem{leong2020}
Paul Hand, Oscar Leong, and Vladislav Voroninski.
\newblock Compressive phase retrieval: Optimal sample complexity with deep generative priors.
\newblock {\em arXiv preprint arXiv:2008.10579}, 2020.

\bibitem{hand2019globalIEEE}
Paul Hand and Vladislav Voroninski.
\newblock Global guarantees for enforcing deep generative priors by empirical risk.
\newblock {\em IEEE Transactions on Information Theory}, 66(1):401--418, 2019.

\bibitem{heckel2021rate}
Reinhard Heckel, Wen Huang, Paul Hand, and Vladislav Voroninski.
\newblock Rate-optimal denoising with deep neural networks.
\newblock {\em Information and Inference: A Journal of the IMA}, 10(4):1251--1285, 2021.

\bibitem{huang2021provably}
Wen Huang, Paul Hand, Reinhard Heckel, and Vladislav Voroninski.
\newblock A provably convergent scheme for compressive sensing under random generative priors.
\newblock {\em Journal of Fourier Analysis and Applications}, 27(2):1--34, 2021.

\bibitem{joshi2021plugin}
Babhru Joshi, Xiaowei Li, Yaniv Plan, and Ozgur Yilmaz.
\newblock Plugin: A simple algorithm for inverting generative models with recovery guarantees.
\newblock {\em Advances in Neural Information Processing Systems}, 34, 2021.

\bibitem{joshi2021plugincs}
Babhru Joshi, Xiaowei Li, Yaniv Plan, and Ozgur Yilmaz.
\newblock {PLUGI}n-{CS}: A simple algorithm for compressive sensing with generative prior.
\newblock In {\em NeurIPS 2021 Workshop on Deep Learning and Inverse Problems}, 2021.

\bibitem{lei2019inverting}
Qi~Lei, Ajil Jalal, Inderjit~S Dhillon, and Alexandros~G Dimakis.
\newblock Inverting deep generative models, one layer at a time.
\newblock {\em Advances in neural information processing systems}, 32, 2019.

\bibitem{ma2018invertibility}
Fangchang Ma, Ulas Ayaz, and Sertac Karaman.
\newblock Invertibility of convolutional generative networks from partial measurements.
\newblock {\em Advances in Neural Information Processing Systems}, 31, 2018.

\bibitem{mardani2018deep}
Morteza Mardani, Enhao Gong, Joseph~Y Cheng, Shreyas~S Vasanawala, Greg Zaharchuk, Lei Xing, and John~M Pauly.
\newblock Deep generative adversarial neural networks for compressive sensing mri.
\newblock {\em IEEE transactions on medical imaging}, 38(1):167--179, 2018.

\bibitem{matousek2013lectures}
Jiri Matousek.
\newblock {\em Lectures on discrete geometry}, volume 212.
\newblock Springer Science \& Business Media, 2013.

\bibitem{menon2020pulse}
Sachit Menon, Alexandru Damian, Shijia Hu, Nikhil Ravi, and Cynthia Rudin.
\newblock Pulse: Self-supervised photo upsampling via latent space exploration of generative models.
\newblock In {\em Proceedings of the ieee/cvf conference on computer vision and pattern recognition}, pages 2437--2445, 2020.

\bibitem{mosser2020stochastic}
Lukas Mosser, Olivier Dubrule, and Martin~J Blunt.
\newblock Stochastic seismic waveform inversion using generative adversarial networks as a geological prior.
\newblock {\em Mathematical Geosciences}, 52(1):53--79, 2020.

\bibitem{ongie2020deep}
Gregory Ongie, Ajil Jalal, Christopher~A Metzler, Richard~G Baraniuk, Alexandros~G Dimakis, and Rebecca Willett.
\newblock Deep learning techniques for inverse problems in imaging.
\newblock {\em IEEE Journal on Selected Areas in Information Theory}, 1(1):39--56, 2020.

\bibitem{pan2021do}
Xingang Pan, Bo~Dai, Ziwei Liu, Chen~Change Loy, and Ping Luo.
\newblock Do 2d {\{}gan{\}}s know 3d shape? unsupervised 3d shape reconstruction from 2d image {\{}gan{\}}s.
\newblock In {\em International Conference on Learning Representations}, 2021.

\bibitem{qiu2020robust}
Shuang Qiu, Xiaohan Wei, and Zhuoran Yang.
\newblock Robust one-bit recovery via relu generative networks: Near-optimal statistical rate and global landscape analysis.
\newblock In {\em International Conference on Machine Learning}, pages 7857--7866. PMLR, 2020.

\bibitem{song2019surfing}
Ganlin Song, Zhou Fan, and John Lafferty.
\newblock Surfing: Iterative optimization over incrementally trained deep networks.
\newblock {\em Advances in Neural Information Processing Systems}, 32, 2019.

\bibitem{sun2016complete}
Ju~Sun, Qing Qu, and John Wright.
\newblock Complete dictionary recovery over the sphere i: Overview and the geometric picture.
\newblock {\em IEEE Transactions on Information Theory}, 63(2):853--884, 2016.

\bibitem{sun2018geometric}
Ju~Sun, Qing Qu, and John Wright.
\newblock A geometric analysis of phase retrieval.
\newblock {\em Foundations of Computational Mathematics}, 18(5):1131--1198, 2018.

\bibitem{vershynin2018high}
Roman Vershynin.
\newblock {\em High-dimensional probability: An introduction with applications in data science}, volume~47.
\newblock Cambridge university press, 2018.

\bibitem{zhu2016generative}
Jun-Yan Zhu, Philipp Kr{\"a}henb{\"u}hl, Eli Shechtman, and Alexei~A Efros.
\newblock Generative visual manipulation on the natural image manifold.
\newblock In {\em European conference on computer vision}, pages 597--613. Springer, 2016.

\end{thebibliography}

\newpage 
\appendix
\addcontentsline{toc}{section}{Appendix} 
\parttoc 

\section{Roadmap}

In Appendix \ref{appx:Props} we establish the main consequences of the \RRWDC that are then used 
to prove Theorem \ref{thm:mainCS_det}. 
Then in Appendix \ref{sec:ProofR2WDC}, we prove Lemma \ref{prop:RandRRWDC} showing that a Gaussian generative network satisfies the \RRWDC with high probability. In Appendix \ref{appx:noise} 
we analyze the perturbation of the gradient 
and objective function due to the noise term $\eta$, and provide the proof of Lemma \ref{lemma:noise1}. 
Extension of the recovery guarantees for  Phase Retrieval, Denosing, 
and Spiked Matrix Recovery are discussed in Appendices \ref{appx:PR}, \ref{appx:Den}, \ref{appx:Spike} respectively. 
Finally, in Appendix \ref{appx:example} we give an example of a network with contractive layers, satisfying the assumptions of our main theorems.


\section{Notation}


For any vector $x$ we denote with $\|x\|$ its Euclidean norm  and for any matrix  $A$ we denote with $\|A\|$ its spectral norm and with $\|A\|_F$ its Frobenius norm. The euclidean inner product between two vectors $a$ and $b$ is $\langle a,b \rangle$.
For a set $S$ we will write $|S|$ for its cardinality and $S^c$ for its complement.
Let $\mathcal{B}(x,r)$ be the Euclidean ball of radius $r$ centered at $x$, and $\mathcal{S}^{k-1}$ be the unit sphere in $\R^k$.
We will use $a = b + O_1(\delta)$ when 
$\|a - b\| \leq \delta$, where the norm is understood to be the absolute value for scalars, the Euclidean norm for vectors and the spectral norm for matrices.
\section{Consequences of the \RRWDC}\label{appx:Props}

Following \cite{hand2019globalIEEE}, we define
the function $g:[0,\pi] \to \R$ which describes how the operator $x \mapsto W_{+,x}$  distorts angles:
\begin{equation}\label{eq:def_g}
	g(\theta) := \cos^{-1}\big( \frac{(\pi - \theta) \cos \theta + \sin \theta}{\pi} \big).
\end{equation}
For two nonzero vectors $x, y$ we let $\bar{\theta}_0 = \angle(x,y)$ and define inductively $\bar{\theta}_i := g(\bar{\theta}_{i-1})$. Then we set 
\begin{equation}\label{eq:htilde}
	\htilde_{x,y} := \frac{1}{2^d} \Bigg[ \big(\prod_{i=0}^{d-1} \frac{\pi - \bar{\theta}_i}{\pi} \big) y + \sum_{i=1}^{d-1} \frac{\sin \bar{\theta}_i }{\pi}   \big(\prod_{j=i+1}^{d-1} \frac{\pi - \bar{\theta}_j}{\pi}\big) {\|y\|} \hat{x} \Bigg].
    \end{equation}

\begin{prop}\label{prop:concGandTheta}
Fix $\eps > 0$ such that $\max(2 d \eps, 10 \eps) < 1$ . Let $G$ be a generative network as in \eqref{eq:Gx} satisfying the \RRWDC  
with constant $\eps$. Then for any $x \in \R^{k}$ and $j \in [d]$
\begin{subequations}
\begin{align} 
  \|x\|^2\Big( \frac{1}{2} - \eps \Big)^j \leq  \|G_j(x)\|^2, \label{eq:boundGjx}
 &\leq   \Big( \frac{1}{2} + \eps \Big)^j \|x\|^2 \\
 \|G(x)\|^2 
 &\leq \frac{1 + 4 \eps d}{2^d} \|x\|^2. \label{eq:boundGx}
\end{align}
\end{subequations}
Moreover, for any $x \neq 0$, $y \neq 0$, $j \in [d]$, the angle $\theta_j = \angle(G_j(x), G_j(y))$ is well-defined and 
\begin{subequations}
\begin{align}
 |\theta_j - g(\theta_{j-1})| &\leq 4 \sqrt{\eps} \label{eq:thetaj}\\
 \langle G(x), G(y) \rangle &\geq \frac{1}{4 \pi} \frac{1}{2^d} \|x\| \|y\| \label{eq:LowProdGxGy} \\
 |\langle G(x), G(y) \rangle - \langle x, \htilde_{x,y} \rangle | &\leq 24 \frac{d^3 \sqrt{\eps}}{2^d}   \|x\| \|y\| \label{eq:concntrGxGy}
\end{align}
\end{subequations}
where $g$ is given in \eqref{eq:def_g} and $\htilde$ in \eqref{eq:htilde}.
\end{prop}

The next result is used to prove concentration of the gradient of 
the objective function around its expectation.
\begin{prop}\label{prop:grad_conctr}
    Fix $0 < \eps < d^{-4}/(16 \pi)^2$ and $d \geq 2$. Suppose that $G$ as in \eqref{eq:Gx} satisfies the \RRWDC with constant $\eps$.  
    Let $x \in \R^k$ be a point where $G(x)$ is differentiable, and
$y \in \R^{k}\setminus \{0\}$, then 
    \begin{align}
            \|\Lambda_{d,x} \|^2 &\leq  \frac{1 + 4 \eps d}{2^d} \leq \frac{13}{12} \frac{1}{2^d} \label{eq:Mdnorm} \\
            \| \Lambda_{d,x}^t\Lambda_{d,x}  - \frac{1}{2^d} I_k  \| &\leq \frac{ 4 \eps d}{2^d} \label{eq:LxTLx} \\
            \|\Lambda_{d,x}^t \Lambda_{d,y} y -  \htilde_{x,y}\| &\leq 24 \frac{d^3 \sqrt{\eps}}{2^d} \|y\| \label{eq:LxTLy}
    \end{align}
\end{prop}

The next proposition uses the \RRWDC to bound the local Lipschitz constant of the $\relu$-networks $\{G_j\}_{j \in [d]}$.
\begin{prop}\label{prop:Lip}
    Suppose that $x \in \mathcal{B}(x, d \sqrt{\eps} \|y\|)$ and $G$ 
    satisfies the \RRWDC with $\eps < 1/(200)^4/d^6$. Then for every $i \in [d]$, it holds that
    \begin{equation}\label{eq:Glip}
        \|G_{i}(x) - G_{i}(y) \| \leq \frac{1.2}{2^{i/2}} \| x - y\|
    \end{equation}
\end{prop}

The next proposition is used to show that when $x$ is close to $y$, 
the gradient of the objective function points in a direction that decreases the distance between of $x$ and $y$.

\begin{prop}\label{prop:convx}
    Suppose $x \in \mathcal{B}(y, d \sqrt{\eps} \|y\|)$ is a differentiable point for $G$, and the \RRWDC holds with $\eps < 1/(200)^4/d^6$. Then it holds that
    \begin{equation}\label{eq:convx}
        \Lambda_x^T (\Lambda_x x - \Lambda_y y ) = \frac{1}{2^d} (x - y) + \frac{1}{2^d} \frac{1}{16} \| x - y\| O_1(1)
    \end{equation}
\end{prop}

We can now prove Theorem \ref{thm:mainCS_det}. 

\begin{proof}[Proof of Theorem \ref{thm:mainCS_det}]

	The proof of Theorem 3.1 in \cite{huang2021provably} only uses the inequality \eqref{eq:boundGjx}-\eqref{eq:convx}, which are proved for a network satisfying the WDC. The previous propositions have shown that such inequalities hold under the weaker \RRWDCs. Therefore from the proof of Theorem 3.1  in \cite{huang2021provably} combined with the Propositions \ref{prop:concGandTheta}-\ref{prop:convx}, we obtain automatically the proof of Theorem \ref{thm:mainCS_det}.
\end{proof}

\subsection{Supplemental Results for Section \ref{appx:Props}}\label{appx:suppProp}

\subsubsection*{Proof of Proposition \ref{prop:concGandTheta}}

\begin{proof}
For $x, y \in \R^k$ and $j \in [d]$, below we write $x_j := G_j(x)$ and $y_j := G_j(y)$.
\medskip

\noindent - \textit{Proof of \eqref{eq:boundGjx}}\\
    Notice that by \eqref{eq:RRWDC}  for $x \in \R^k$
    \[
        \Big(\frac{1}{2} - \eps \Big)  \| x_{j-1} \|^2 \leq \|x_j\|^2  \leq \Big(\frac{1}{2} + \eps \Big)  \| x_{j-1} \|^2,
    \]
    which proceeding by induction gives \eqref{eq:boundGjx}. 
    \medskip
    
    \noindent - \textit{Proof of \eqref{eq:boundGx}}\\
    Next observe that since $\log (1 + z)\leq z$, $e^z \leq 1 + 2z$ for $z < 1$ and $2 d \eps \leq 1$, from \eqref{eq:boundGjx} we have 
    \[
        \|G_d(x)\|^2 \leq \frac{(1 + 2 \eps)^d}{2^d} \|x\|^2 \leq \frac{1}{2^d} e^{d \log (1 + 2 \eps)} \|x\| \leq \frac{1 + 4 \eps d}{2^d} \|x\|,
    \]
    which corresponds to \eqref{eq:boundGx}.
    \medskip
    
    \noindent - \textit{Proof of \eqref{eq:thetaj}}\\
    Assume that $x, y \in \R^{k} \setminus \{0\}$. Then, the assumption $2 d \eps \leq 1$ 
    and the lower bound in \eqref{eq:boundGjx} imply that $\theta_j$ are well-defined for all $j \in [d]$. 
    To prove then \eqref{eq:thetaj} notice that it is sufficient to prove that for any $j \in [d]$ it holds that
    \[
        \Big| \cos \theta_j - \frac{(\pi - \theta_{j-1}) \cos \theta_{j-1} + \sin \theta_{j-1}}{\pi} \Big| \leq 5 \eps 
    \]
    By homogeneity of the $\relu$ activation function, we can assume without loss of generality that $\| x_{j-1} \| = \| y_{j-1} \| = 1$. Let  
    \[
        \begin{aligned}
            \delta_1 &:= \langle x_{j-1}, \big( \Wjpx^T \Wjpy - Q_{x_{j-1}, y_{j-1}}\big) y_{j-1} \rangle  \\ 
            \delta_2 &:= \langle x_{j-1}, \big( \Wjpx^T \Wjpx - I_k/2\big) y_{j-1} \rangle  \\ 
            \delta_3 &:= \langle y_{j-1}, \big( \Wjpy^T \Wjpy - I_k/2\big) y_{j-1} \rangle  \\ 
        \end{aligned}
    \]
    and notice that by the \RRWDC we have $\max(|\delta_1|, |\delta_2|, |\delta_3|) \leq \eps$. Thus,
    \[
    \begin{aligned}
        \cos \theta_{j} 
        &= \frac{\langle x_{j}, y_{j}\rangle}{\|x_{j}\| \|y_{j}\|} \\
        &= \frac{\langle x_{j-1}, \Wjpx^T \Wjpy y_{j-1} \rangle}{\sqrt{\langle x_{j-1}, \Wjpx^T \Wjpx x_{j-1} \rangle \langle y_{j-1}, \Wjpy^T \Wjpy y_{j-1} \rangle}}
        \\
        &= 2 \frac{\langle x_{j-1}, Q_{x_{j-1}, y_{j-1}} y_{j-1} \rangle + \delta_1}{\sqrt{(1 + 2 \delta_2)(1 + 2 \delta_3)}}.
    \end{aligned}
    \]
    Finally, notice that $2 \langle x_{j-1}, Q_{x_{j-1}, y_{j-1}} y_{j-1} \rangle = [(\pi - \theta_{j-1}) \cos \theta_{j-1} + \sin \theta_{j-1}]/\pi$ so 
    \begin{align*}
        |\cos \theta_j - 2 \langle x_{j-1}, Q_{x_{j-1}, y_{j-1}} y_{j-1} \rangle| &\leq 2 |\langle x_{j-1}, Q_{x_{j-1}, y_{j-1}} y_{j-1} \rangle| \Bigg| 1 - \frac{1}{\sqrt{(1 + 2 \delta_2)(1 + 2 \delta_3)}} \Bigg| \\
        &\qquad + \frac{2 | \delta_1 |}{\sqrt{(1 + 2 \delta_2)(1 + 2 \delta_3)}}\\
        &\leq \Big| 1 - \frac{1}{(1 - 2 \eps)} \Big| + \frac{2 \eps }{ (1 - 2 \eps)}\\
        &\leq 5 \eps
    \end{align*}
     where the second inequality follows from $|2 \langle x_{j-1}, Q_{x_{j-1}, y_{j-1}} y_{j-1} \rangle| \leq 1$ 
     and $\max(|\delta_1|, |\delta_2|, |\delta_3|) \leq \eps$, and the third inequality from $10 \eps < 1$.
     \medskip
    
    \noindent - \textit{Proof of \eqref{eq:LowProdGxGy}}\\
    By \eqref{eq:boundGjx} and $\eps \leq 1/2$, it follows that $\|x_d\| \|y_d\| \geq \frac{(1 - 2 \eps)^d}{2^d} \|x\| \|y\| \frac{1 - 2 d \eps}{2^d}$.
    Moreover, let $\delta := 4 \sqrt{\eps}$, 
    then by \eqref{eq:thetaj} we have that $\theta_j = g(\theta_{j-1}) + O_1(\delta)$. 
    Thus, $\theta_d = g(g(\cdots g(g(\theta_0) + O_1(\delta)) + O_1(\delta) \cdots) +O_1(\delta))) + O_1(\delta)$ 
    and for $\bar{\theta} = g^{\circ d}(\theta_0)$, so that, using $g'(\theta) \leq 1$ for all $\theta$, we have 
    \begin{equation}\label{eq:thetabard}
        |\theta_d - \bar{\theta}_d | \leq d \delta. 
    \end{equation}
    Then by \eqref{eq:thetabard}, $\bar{\theta}_d \leq \cos^{-1}({1}/{\pi})$ for $d\geq 2$, and $16 \pi d \sqrt{\eps} < 1$, 
    follows that $\cos \theta_d \geq 3/(4 \pi)$.  
    
    Finally, if $2 d \eps \leq 2/3$, we can then conclude that 
    \[
        \langle G(x), G(y)\rangle \geq \cos (\theta_d) \|x_d\| \|y_d\| \geq \frac{1}{4 \pi} \frac{1}{2^d} \|x\| \|y\|.
    \]
    \medskip
    
    \noindent - \textit{Proof of \eqref{eq:concntrGxGy}}\\
    The following result on a recurrence relation will be used in the subsequent analysis
    \begin{equation}\label{eq:Gammad}
        \Gamma_d = s_d \Gamma_{d-1} + r_d, \quad \Gamma_0 = y \implies \Gamma_d = \Big( \prod_{i=1}^d s_i \Big) y + \sum_{i=1}^d \Big( r_i \prod_{j=i+1}^d s_j \Big)
    \end{equation}
    Define $\Gamma_d := \langle x_d, y_d \rangle$, then
    \begin{subequations}
    \begin{align*}
        \Gamma_d &= \langle x_{d-1}, W_{d-1,+,x}^T W_{d-1,+,x} y_{d-1} \rangle \\
        &= \langle x_{d-1}, Q_{x_{d-1}, y_{d-1}} y_{d-1} \rangle + O_1(\eps) \|y_{d-1}\| \|x_{d-1}\|, \\
        &= \frac{\pi - \theta_{d-1}}{2 \pi} \Gamma_{d-1} + \frac{\sin \theta_{d-1}}{2 \pi} \|x_{d-1}\| \|y_{d-1} \| + O_1(\eps) \|y_{d-1}\| \|x_{d-1}\|, \\
        &= \frac{\pi - \theta_{d-1}}{2 \pi} \Gamma_{d-1} + \frac{\sin \theta_{d-1}}{2 \pi} \frac{\|x\| \|y\|}{2^{d-1}} 
        +  \frac{\eps}{2^{d}} \Big(\frac{4 \eps d}{\pi}+ 2(1 + 4 \eps d)\Big) \|y\| \|x\|O_1(1),\\
        &= \frac{\pi - \theta_{d-1}}{2 \pi} \Gamma_{d-1} + \frac{\sin \theta_{d-1}}{2 \pi} \frac{\|x\| \|y\|}{2^{d-1}} 
        +  11 d \eps \frac{\|y\| \|x\|}{2^d}O_1(1),
    \end{align*}
    \end{subequations}
    Where the second equality follows from the \RRWDC, the third from the definition of $Q_{p,q}$.  
    The rest of the proof proceeds as in the proof of Lemma 8 in \cite{hand2019globalIEEE}.
\end{proof}

\subsubsection*{Proof of Proposition \ref{prop:grad_conctr}}
\begin{proof}
\smallskip
${}$

\noindent - \textit{Proof of \eqref{eq:LxTLx}}.\\
Let $x \in \R^k$ be a point where $G$ is differentiable, and notice that 
for small enough $z$, by local linearity of $G$, we have $G(x + z) = \Lambda_x z$. Then the \RRWDC gives for $j \in [d]$
\begin{equation*}%
     \big| \langle \big( \Wjpx^T \Wjpy - I_k/2 \big) \Lambda_{j-1,x} z ,  \Lambda_{j-1,x} z \rangle \big| \leq \eps \|\Lambda_{j-1,x} \|^2 \|z\|^2
\end{equation*}
for all $z$, which in turn implies
\begin{equation}\label{eq:RRWDC_Lx}
    \|\Lambda_{j, x}^T \Lambda_{j, x} - \frac{1}{2}\Lambda_{j-1, x}^T \Lambda_{j-1, x} \| \leq \eps \|\Lambda_{j-1, x}^T \Lambda_{j-1, x} \|.
\end{equation}

Let now $M_d := \Lambda_{d,x}^T \Lambda_{d,x}$ with $M_0 = I_{k}$, then  
\begin{equation}\label{eq:Md}
    M_d = \frac{1}{2} M_{d-1} + \|M_{d-1}\| O_1(\eps). 
\end{equation}
We then obtain
\begin{equation*}
    \| M_d \| \leq \Big( \frac{1}{2} + \eps \Big) \| M_{d-1}\| \leq  \frac{(1 + 2 \eps)^d}{2^d} \| M_0 \| \leq \frac{1 + 4 \eps d}{2^d},
\end{equation*}
where the second inequality, and the third inequality uses $2 d \eps \leq 1$ and the same reasoning as in the proof of \eqref{eq:boundGx}.
From \eqref{eq:Mdnorm} and \eqref{eq:Md}  we obtain the following recurrence relation
\[
M_d 
= \frac{1}{2} M_{d-1} + O_1 \Big( \eps \frac{1 + 4 \eps (d-1)}{2^{d-1}} \Big),
\]
which, using \eqref{eq:Gammad} and $4 \eps d \leq 1$, gives 
\[
\begin{aligned}
M_d &= \frac{1}{2} I_k + \sum_{i=1}^d O_1 \Big( \eps \frac{1 + 4 \eps (i-1)}{2^{i-1}} \Big) \frac{1}{2^{d-i}} \\
&= \frac{1}{2} I_k + \frac{4 \eps d}{2^d} O_1(1)
\end{aligned}
\]
\medskip

\noindent - \textit{Proof of \eqref{eq:LxTLy}}.\\
Notice again that if $x \in \R^k$ is a differentiable point for $G$, the \RRWDC gives for any $j \in [d]$
\begin{equation}\label{eq:RRWDC_Lxy}
    \|\Lambda_{j, x}^T \Lambda_{j, y} y  - \Lambda_{j-1, x}^T Q_{x_{j-1}, y_{j-1}} \Lambda_{j-1, y} y  \| \leq \eps \|\Lambda_{j-1, x} \| \|G_{j-1}(y) \|.
\end{equation}
We then let $\Gamma_d := \Lambda_{d, x}^T \Lambda_{d, y} y $ and observe that
\[
\begin{aligned}
    \Gamma_d 
    &= 
    \Lambda_{d-1, x}^T Q_{x_{d-1}, y_{d-1}} \Lambda_{d-1, y} y + \|\Lambda_{d-1, x} \| \|G_{d-1}(y) \| O_1(\eps) \\
    &= \frac{\pi - \theta_{d-1}}{2 \pi} \Gamma_{d-1} + \frac{\sin \theta_{d-1}}{2 \pi} \frac{\| y_{d-1} \|}{\|x_{d-1}\|} \Lambda_{d-1,x}^T \Lambda_{d-1,x} x + \eps \Big( \frac{1 + 4 \eps d}{2^{d-1}} \Big) \|y\|
\end{aligned}
\]
where the first equality is from \eqref{eq:RRWDC_Lxy}, and the second uses the definition of $Q_{x,y}$,  \eqref{eq:boundGx} and \eqref{eq:Mdnorm}. 
The rest of the proof follows as in the proof of Equation (7) in Lemma 8 in \cite{hand2019globalIEEE}.
\end{proof}

\subsubsection*{Proof of Proposition \ref{prop:Lip}}

\begin{lemma}\label{lemma:GjxmGjy}
Suppose $G$ satisfies the \RRWDC with constant $\eps$. Then for any $x, y \in \R^{k} \setminus \{0\}$ and $i \in [d]$, it holds that
\[
\|G_i(x) - G_i(y)\| \leq \Bigg( \sqrt{\frac{1}{2} + \eps} + \sqrt{2 (2 \eps + \theta_{i-1})} \Bigg) \|G_{i-1}(x) - G_{j-1}(y)\|
\]
where $\theta_{i-1} = \angle(G_i(x), G_i(y))$.
\end{lemma}

\begin{proof}[Proof of Lemma \ref{lemma:GjxmGjy}]
We have
\begin{equation}\label{eq:GjxGjy0}
    \|G_j(x) - G_j(y)\| \leq  \| (W_j)_{+,x_{j-1}} (x_{j-1} - y_{j-1}) \| + \|\big( \Wjpx - \Wjpy \big) y_{j-1}\|. 
\end{equation}
We begin analyzing the first term, noticing that by the \RRWDC
\begin{align}
 \| \Wjpx (x_{j-1} - y_{j-1})\|^2 &= (x_{j-1} - y_{j-1})^T (\Wjpx^T \Wjpx - \frac{1}{2} I_{n_1}) (x_{j-1} - y_{j-1})  + \frac{1}{2} \|x_{j-1} - y_{j-1}\|^2 \nonumber \\
 &\leq \Big(\frac{1}{2} + \eps \Big) \|x_{j-1} - y_{j-1} \|^2 \label{eq:GjxGjy1}
\end{align}
We next analyze the second term. Let $W_{j,i} \in \R^{1 \times n_{j-1}}$ be the $i$-th row of $W_{j}$ then
\begin{align}
    \| \big( \Wjpx - \Wjpy \big) y_{j-1} \|^2
    &= \sum_{i=1}^n \big( \indx_{\Wji  x_{j-1} > 0} - \indx_{\Wji  y_{j-1} > 0}\big)^2  (\Wji \,  y_{j-1})^2 \nonumber \\
    &\leq \sum_{i=1}^n \big( \indx_{\Wji  x > 0} - \indx_{\Wji  y > 0} \big)(\Wji  (x_{j-1} - y_{j-1}) \nonumber \\
    &= \sum_{i=1}^n \indx_{\Wji  x > 0} \indx_{\Wji  y \leq 0} \Wji  (x_{j-1} - y_{j-1}) \nonumber \\ 
    &\quad+ \sum_{i=1}^n \indx_{\Wji  x \leq 0} \indx_{\Wji  y > 0} \Wji  (x_{j-1} - y_{j-1}) \nonumber \\
    &= (x_{j-1} - y_{j-1})^T (W_j)_{+,x_{j-1}}^T \Big( (W_j)_{+,x_{j-1}} - (W_j)_{+,y_{j-1}} \Big) (x_{j-1} - y_{j-1}) \nonumber \\
    &\quad+ (x_{j-1} - y_{j-1})^T (W_j)_{+,y_{j-1}}^T \Big( (W_j)_{+,y_{j-1}} - (W_j)_{+,x_{j-1}} \Big) (x_{j-1} - y_{j-1}). \label{eq:GjxGjy2}
\end{align}
Observe now that by the \RRWDC we have
\begin{align*}
    |(x_{j-1}-y_{j-1})^T &(W_j)_{+,y_{j-1}}^T \Big( (W_j)_{+,y_{j-1}} - (W_j)_{+,x_{j-1}} \Big) (x_{j-1} - y_{j-1})|  \nonumber \\  
    &\leq |(x_{j-1}-y_{j-1})^T  \Big( (W_j)_{+,y_{j-1}}^T (W_j)_{+,y_{j-1}} - I_k/2  \Big) (x_{j-1} - y_{j-1})| \nonumber \\
    &+ |(x_{j-1}-y_{j-1})^T  \Big( (W_j)_{+,y_{j-1}}^T (W_j)_{+,x_{j-1}} - Q_{x_{j-1}, y_{j-1}}  \Big) (x_{j-1} - y_{j-1})| \nonumber \\
    &\,+ |(x_{j-1}-y_{j-1})^T  \Big( I_{n_{i-1}}/2 - Q_{x_{j-1}, y_{j-1}}  \Big) (x_{j-1} - y_{j-1})| \nonumber\\ 
    &\leq (2 \eps + \theta_{j-1})  \| x_{j-1} - y_{j-1}\|^2,
\end{align*}
which together with \eqref{eq:GjxGjy2} gives
\begin{equation}\label{eq:huangTrick}
    \| \big( \Wjpx - \Wjpy \big) y_{j-1} \|^2 \leq 2 (2 \eps + \theta_{j-1})  \| x_{j-1} - y_{j-1}\|^2.
\end{equation}
We conclude using \eqref{eq:GjxGjy1} and \eqref{eq:huangTrick} in \eqref{eq:GjxGjy0}. 
\end{proof} 

\subsubsection*{Proof of Proposition \ref{prop:convx}}

We next prove the convexity-like property in Proposition \ref{prop:convx}.

\begin{proof}[Proof of Proposition \ref{prop:convx}]

    We begin observing that by \eqref{eq:thetabard} we have $| \theta_i - \bar{\theta}_i |\leq 4 i \sqrt{\eps} \leq 4 d \sqrt{\eps}$. Furthermore, 
    since $x \in \mathcal{B}(y, d \sqrt{\eps} \|y\|)$ it follows that 
    \[
    \bar{\theta}_i \leq \bar{\theta}_{0} \leq 2 d \sqrt{\eps}.
    \]
    Thus by the assumption on $\eps$, we have
    \begin{equation}\label{eq:boundthetieps}
        \sqrt{2} \sqrt{\theta_i + 2 \eps} 
        \leq \sqrt{2} \sqrt{\bar{\theta}_i + 4 d \sqrt{\eps} + 2 \eps}
        \leq \sqrt{2} \sqrt{2 d \sqrt{\eps} + 4 d \sqrt{\eps} + 2 \eps} \leq \frac{1}{30 \sqrt{2} d} 
    \end{equation}
    
    Let now $\Gamma_d := \Lambda_{d,x}^T (\Lambda_{d,x} x - \Lambda_{d,y} y )$. Then notice that
    \begin{align}
        \Gamma_d &= \Lambda_{d-1,x}^T W_{d,+,x}^T  ( W_{d, +, x} \Lambda_{d-1,x} x - W_{d, +, y} \Lambda_{d-1,y} y ) \nonumber \\
        &=  \Lambda_{d-1,x}^T W_{d,+,x}^T W_{d, +, x} (  \Lambda_{d-1,x} x -  \Lambda_{d-1,y} y ) + \Lambda_{d,x}^T  (W_{d, +, x} - W_{d, +, y}) \Lambda_{d-1,y} y \nonumber \\
        &= \frac{1}{2} \Gamma_{d-1} + \eps \|\Lambda_{d-1,x}\| \|\Lambda_{d-1,x} x -  \Lambda_{d-1,y} y\| O_1(1) + \|\Lambda_{d,x}\| \| (W_{d, +, x} - W_{d, +, y}) \Lambda_{d-1,y} y\| O_1(1) \nonumber\\
        &= \frac{1}{2} \Gamma_{d-1} + \Big( \eps + \sqrt{\frac{1}{2} + \eps} \sqrt{2 (2 \eps + \theta_{d-1})}\Big) \|\Lambda_{d-1}\| \|\Lambda_{d-1,x} x -  \Lambda_{d-1,y} y\| O_1(1) \nonumber \\
        &= \frac{1}{2} \Gamma_{d-1} + \Big( \eps + \sqrt{\frac{1}{2} + \eps} \sqrt{2 (2 \eps + \theta_{d-1})}\Big) \frac{1.2 \sqrt{1 + 4 \eps d}}{2^{d-1}} \|x - y\| O_1(1) \nonumber \\
        &= \frac{1}{2} \Gamma_{d-1} + 2 \Big( \frac{1}{200^4 d^6} + \frac{1}{30 \sqrt{2} d } \Big) \frac{\|x - y\|}{2^{d-1}}  O_1(1) \label{eq:recurrConvx}
    \end{align}
    where the third equality follows from the \RRWDCs, the fourth the \RRWDC and \eqref{eq:huangTrick}, the fifth from \eqref{eq:Mdnorm} 
    and Proposition \ref{prop:Lip}, and sixth from \eqref{eq:boundthetieps} and the assumption on $\eps$.
    Finally, from \eqref{eq:recurrConvx} and \eqref{eq:Gammad}  we obtain
    \[
        \Gamma_d = \frac{1}{2^d} \|x - y\| + \frac{1}{16} \frac{\|x - y\|}{2^d} O_1(1)
    \]
\end{proof}

\section{Proof of  Lemma \ref{prop:RandRRWDC}}\label{sec:ProofR2WDC}

In this section, we prove that a generative network $G$ with random weights satisfies the \RRWDC with high-probability 
(Lemma \ref{prop:RandRRWDC}). Our proof is inspired by the proof of Proposition 3 in \cite{leong2020}. 
\smallskip

Notice that because of the piecewise-linear nature of the $\relu$ activation function, 
the output of a $\relu$ network is a subset of a union of affine subspaces. 
The following lemma from \cite{joshi2021plugin} provides an upper bound on the number of such subspaces.
\begin{lemma}[Lemma 7 in \cite{joshi2021plugin}]\label{lemma:joshiRange}
Consider a generative network $G$ as in \eqref{eq:Gx} and assume that  $n_i \geq k$ for $i \in [d]$. 
Then for $i \in [d]$, range$(G_i)$ is contained in a union of affine subspaces. Precisely,
\[
	\text{range}(G_i) \subseteq \cup_{j \in [\Psi_i]} S_{i,j}\quad \text{where} \quad \Psi_i \leq \prod_{j=1}^i \Big( \frac{e n_j}{k}\Big)^k.
\]
Here each $S_{i,j}$ is some $k$-dimensional affine subspace (which depends on $\{W_{\ell}\}_{\ell \in [i]}$) in $\R^{n_i}$.
\end{lemma}

We next give the main result upon which the proof of Proposition \ref{prop:RandRRWDC} rests.
\begin{prop}\label{prop:unifSubs}
Fix $0 < \eps < 1$ and $\ell < n$. Let $W \in \R^{m \times n}$ have i.i.d. $\mathcal{N}(0, 1/m)$. 
Let $R,S$ be $\ell$-dimensional subspaces of $\R^n$, and $T$ be an $\ell'$-dimensional 
subspaces of $\R^n$ with $l' \geq l$. Then if $m \geq C_\eps \ell'$, we have that 

\begin{equation}\label{eq:Subs_uv}
    | \langle \Wpr^T \Wps u, v\rangle - \langle Q_{r,s} u, v \rangle  | \leq \eps \| u \|_2 \|v\|_2 \quad \forall \; u, v \in T, \, \forall r \in R, \, \forall s \in S, 
\end{equation}
with probability exceeding
\[
	1 - \gamma \Big(\frac{e \, m}{\ell}\Big)^{2 \ell} \exp(- {c}_\eps m)
\]
Furthermore, let $U = \bigcup_{i = 1}^{N_1} U_i$, $V = \bigcup_{j = 1}^{N_2} V_j$ $V = \bigcup_{j = 1}^{N_2} V_j$ , $R = \bigcup_{p = 1}^{N_3} R_p$, 
and $S = \bigcup_{q = 1}^{N_4} S_q$ be union of subspaces of $\R^n$ of dimension at most $\ell$. Then if $m \geq 2 C_\eps \ell'$
\begin{equation}\label{eq:Subs_UV}
    | \langle \Wpr^T \Wps u, v\rangle - \langle Q_{r,s} u, v \rangle  | \leq \eps \| u \|_2 \|v\|_2 \quad \forall \; u \in U,\, v \in V, \, \forall r \in R, \, \forall s \in S, 
\end{equation}
with probability exceeding
\[
	1 - \gamma N_1 N_2 N_3 N_4 \Big(\frac{e \, m}{\ell}\Big)^{2 \ell} \exp(- {c}_\eps m).
\]
 Here $c_\eps$ depends polynomially on $\eps$, $C_\eps$ depends polynomially on $\eps^{-1}$, 
 and $\gamma$ is a positive universal constant. 
\end{prop} 

With the above two results, we are in a position to prove Lemma \ref{prop:RandRRWDC}.
\begin{proof}[Proof of Lemma \ref{prop:RandRRWDC}]
We begin establishing the proposition in the $d=2$ case. 

If $n_1 \geq 2 C_\eps k$ by the second part of Proposition \ref{prop:unifSubs} with $U, V, R, S = \R^k$, $W_1$ 
satisfies \eqref{eq:RRWDC} with probability at least
\[
1 - \gamma \Big( \frac{e n_1}{k}\Big)^{2k} \exp(- c_\eps n_1).
\]

	We next consider the bound \eqref{eq:RRWDC} for $j=2$.
	 Fix $W_1$ and observe that, by Lemma \ref{lemma:joshiRange}, range$(G_1)$ 
	 is contained in the union of at most $\Psi_1$ number of $k$-dimensional affine 
	 subspaces of $\R^{n_1}$ and $\{ G_1(x_1) - G_1(x_2):  x_1, x_2  \in \R^k \}$ 
	 is contained in the union of at most $\Psi_1^2$ number of $2k$-dimensional affine subspaces of $\R^{n_1}$. 
	 Since then an $\ell$-dimensional affine subspace is also contained in an $\ell+1$ subspace. 
	 We have that range$(G_1) \subset \mathcal{R}_1$  where $\mathcal{R}_1$ is  the union of at most $\Psi_1$ 
	 number of $k+1$-dimensional subspaces and $\{ G_1(x_1) - G_1(x_2):  x_1, x_2  \in \R^k \} \subset \mathcal{U}_1$ 
	 where $\mathcal{U}_1$ is the union of at most $\Psi_1^2$ number of $2k+1$-dimensional  subspaces.
	 
	 By applying the second part of Proposition \ref{prop:unifSubs} 
	 to the sets $\mathcal{U}_1, \mathcal{U}_1, \mathcal{R}_1$ and $\mathcal{R}_1$, we have that for fixed $W_1$,
	 \begin{multline}\label{eq:RRWDCj2}
     \Big| \langle \Big( (W_{2})_{+, G_{1}(x)}^T (W_{2})_{+,G_{1}(y)} - Q_{G_{1}(x),G_{1}(y)} \Big)\big( G_{1}(x_1) - G_{1}(x_2)\big),  G_{1}(x_3) - G_{1}(x_4) \rangle \Big| \\ \leq \eps \|G_{1}(x_1) - G_{1}(x_2) \|_2 \|G_{1}(x_3) - G_{1}(x_4) \|_2
	\end{multline}
	with probability at least 
	\[
		1 - {\gamma} \Psi_1^6 \Big(\frac{e \, n_2}{k+1}\Big)^{2 k +2}  e^{- {c}_\eps n_2} \geq 1 - {\gamma}  \Big(\frac{e \, n_2}{k+1}\Big)^{4 k} e^{- {c}_\eps n_2/2}
	\]
	provided that $n_2 \geq 12 {c}_\eps^{-1} \log \Psi_1 $ and $n_2 \geq 2 C_\eps (2 k + 1)$. 
	In particular the above holds provided that $n_2 \geq \widetilde{C}_\eps k \log({e n_1}/{k})$ 
	where $\widetilde{C}_\eps$ depends polynomially on $\eps^{-1}$.
	
	Integrating over the probability space of $W_1$, independence of $W_2 $ and $W_1$ implies that \eqref{eq:RRWDCj2} 
	holds for random $W_1$ with the same probability bound. 
	This allows us to conclude that a two-layer random generative network $G$ satisfies the \RRWDC with probability at least 
	\[
		1 - \gamma \Big(\frac{e n_1}{k}\Big)^{2 k}  e^{- {c}_\eps n_1 } - {\gamma}  \Big(\frac{e \, n_2}{k+1}\Big)^{4 k} e^{- {c}_\eps n_2/2}.
	\] 
	 
	The proof of the  $d \geq 2$ case follows similarly. 
	In particular, to establish \eqref{eq:RRWDC} for $W_i$ notice that range$(G_{i-1})$ 
	is contained in the union of at most $\Psi_{i-1}$ number $k+1$ subspaces, 
	and $\{ G_{i-1}(x_1) - G_{i-1}(x_2):  x_1, x_2  \in \R^k \}$ in the union of at most $\Psi_{i-1}^2$ 
	number of $2k+1$-dimensional subspaces. Applying Proposition \ref{prop:unifSubs} 
	to these subspaces we have that for fixed $\{W_j\}_{j \in [i-1]}$
	\begin{multline}\label{eq:RRWDCjj}
     \Big| \langle \Big( (W_{i})_{+, G_{i-1}(x)}^T (W_{i})_{+,G_{i-1}(y)} - Q_{G_{i-1}(x),G_{i-1}(y)} \Big)\big( G_{i-1}(x_1) - G_{i-1}(x_2)\big),  G_{i-1}(x_3) - G_{i-1}(x_4) \rangle \Big| \\ \leq \eps \|G_{i-1}(x_1) - G_{i-1}(x_2) \|_2 \|G_{i-1}(x_3) - G_{i-1}(x_4) \|_2
	\end{multline}
	with probability at least
	\[
 		1 - {\gamma}  \Big(\frac{e \, n_i}{k+1}\Big)^{4 k} e^{- {c}_\eps n_i/2}
	\]
	provided that
	\[
		n_i \geq \widetilde{C}_\eps \cdot k \cdot \prod_{j=1}^{i-1}
         \frac{e \,n_j}{k}.
	\]
	Integrating over the probability space of $\{W_j \}_{j \in [i-1]})$ 
	indpendence of $W_i$ and $(W_1, \dots, W_{i-1})$ gives that \eqref{eq:RRWDCjj} holds  with the same probability bound.
\end{proof}

We will devote the following section to the proof of Proposition \ref{prop:unifSubs}. 

\subsection{Proof of Proposition \ref{prop:unifSubs}}

We begin by proving a weaker form of Proposition \ref{prop:unifSubs}, that characterizes the concentration of $\Wpr^T \Wps $ around its mean for fixed $r, s$ and when acting on $\ell$-dimensional subspaces. 

\begin{lemma}\label{lemma:fixSubs}
Fix $0 < \eps < 1$ and $k < n$. Let $W \in \R^{m \times n}$ have i.i.d. $\mathcal{N}(0, 1/m)$ entries and fix $r, s \in \R^{n}$. Let $T$  be a $\ell$-dimensional subspace of $\R^n$. Then if $m \geq \tilde{C}_1 \ell$, we have that with probability exceeding $1 - 2 \exp(- \tilde{c}_1 \,m)$,
\begin{equation}\label{eq:fixSubs_uu}
    | \langle \Wpr^T \Wps u, u\rangle - \langle Q_{r,s} u, u \rangle  | \leq  \eps \| u \|_2^2  \quad \forall \; u \in T
\end{equation}
and 
\begin{equation}\label{eq:fixSubs_uv}
    | \langle \Wpr^T \Wps u, v\rangle - \langle Q_{r,s} u, v \rangle  | \leq 3 \eps \| u \|_2 \|v\|_2 \quad \forall \; u, v \in T,
\end{equation}
Furthermore, let $U = \bigcup_{i = 1}^{N_1} U_i$ and $V = \bigcup_{j = 1}^{N_2} V_j$ where $U_i$ and $V_j$ are subspaces of $\R^n$ of dimension at most $\ell$ for all $i \in [N_1]$ and $j \in [N_2]$. Then if $m \geq 2 \tilde{C}_1 \ell$
\begin{equation}\label{eq:fixSubs_UV}
    | \langle \Wpr^T \Wps u, v\rangle - \langle Q_{r,s} u, v \rangle  | \leq 3 \eps \| u \|_2 \|v\|_2 \quad \forall \; u \in U, \forall \; v \in V,
\end{equation}
with probability exceeding $1 - 2 N_1 N_2 \exp(- \tilde{c}_1 m)$. Here $\tilde{c}_1$ depends polynomially on $\eps$ and $\tilde{C}_1 = \Omega(\eps^{-1} \log \eps^{-1})$. 
\end{lemma}

\begin{proof}
The proof follows the one in Proposition 4 of \cite{leong2020} with minor variations. Set $\Sigma_{r,s} := \Wpr^T \Wps - Q_{r,s}$,  and notice that for fixed $u \in \R^{n-1}$,
$
    \langle \Sigma_{r,s} u, u \rangle = \sum_{i=1}^m Y_i 
$
where $Y_i = X_i - \EX[X_i]$, $X_{i} = \indx_{\langle w_i, r \rangle > 0}\indx_{\langle w_i, s \rangle > 0} \langle w_i, u \rangle^2$ and each $w_i \sim \mathcal{N}(0, I_n/m)$. We then notice that the $Y_i$ are sub-exponential random variables and by standard $\eps$-net argument we can show that \eqref{eq:fixSubs_uu} holds with high-probability.   
Proposition 5 in \cite{leong2020} can then be adapted to this case as well and used to derive \eqref{eq:fixSubs_uv} from \eqref{eq:fixSubs_uu}. Finally \eqref{eq:fixSubs_UV} follows by a union bound over all subspaces of the form span$(U_i, V_j)$.
\end{proof}

We next observe that the rows of a sufficiently tall random matrix $W$ tessellate the unit sphere in regions of small diameter. 

\begin{lemma}\label{lemma:conseqTess}
Fix $0 < \eps < 1$. Let $W \in \R^{m \times n}$ have i.i.d. $\mathcal{N}( 0, 1/m)$ 
entries with rows $\{w_\ell\}_{\ell=1}^m$. Let $Z$ be a $\ell$-dimensional subspace of $\R^n$. 
Define $E_{Z,W}$ to be the event that there exists a set $Z_0 \subset Z$ with the following properties:
\begin{enumerate}[i)]
    \item each $z_0 \in Z_0$ satisfies $\langle w_\ell, z_0 \rangle \neq 0$ for all $\ell \in [m]$,
    
    \item $|Z_0| \leq (\frac{e \, m}{\ell})^\ell$, and 
    
    \item for all $z \in Z$ such that $\|z \|_2 = 1$, there exists $z_0 \in Z_0$ such that $\| z - z_0 \|_2\leq \eps$.
\end{enumerate}

If $m \geq \tilde{C}_2 \ell$, then $\PX(E_{Z,W}) \geq 1 - C_2 \exp(- c_2 \eps m)$. 
Here $C_2$ and $c_2$ are positive absolute constants and $\tilde{C}_2$ depends polynomially on $\eps^{-1}$. 
\end{lemma}
\begin{proof}
	The proof of this lemma follows the one in Lemma 24 in \cite{leong2020}. 
    The upper bound $|Z_0| \leq (\frac{e \, m}{\ell})^\ell$ is due to Lemma \ref{lemma:partition} in Appendix \ref{appx:suppProofRRWDC}.
\end{proof}

We are now ready to present the proof of Proposition \ref{prop:unifSubs}.

\begin{proof}[Proof of Proposition \ref{prop:unifSubs}]

Let $E_{R,W}$ be the event defined in Lemma \ref{lemma:conseqTess} corresponding to the matrix $W$ and subspace $R$. 
On the event $E_{R,W}$ there exists a finite set $R_0 \subset R$ satisfying properties i) - iii) of Lemma \ref{lemma:conseqTess}. 
Similarly, we can define the event $E_{S,W}$ for the matrix $W$ and subspace $S$, and the finite set $S_0 \subset S$ satisfying properties i) - iii).

 We can then define the event $E_{R,S} := E_{R,W} \cap E_{S,W}$ so that if $m \geq \tilde{C}_2 \ell'$ by Lemma  \ref{lemma:conseqTess} we have
\[
    \PX (E_{R,S}) \geq 1 - 2 C_2 \exp(- c_2 \eps m).
\]
 For fixed $r_0 \in R_0$ and $s_0 \in S_0$, Lemma \ref{lemma:fixSubs} gives that if $m \geq 2 \tilde{C}_1 \ell$ 
 with probability at least $1 - 2 \exp(- \tilde{c}_1 m)$
\[
    |\langle W_{+, r_0}^T W_{+, s_0} u, v \rangle - \langle Q_{r_0, v_0} u, v \rangle| \leq 3 \eps \|u\|_2 \|v\|_2 \quad \forall u,v \in T.
\]
Next, let $E_0$ be the event that
\[
|\langle W_{+, r_0}^T W_{+, s_0} u, v \rangle - \langle Q_{r_0, v_0} u, v \rangle| \leq 3 \eps \|u\|_2 \|v\|_2 \quad \forall \, u,v \in T, \, r_0 \in R_0,\,  s_0 \in S_0.
\] 
Then, on $E_{R,S}$, a union bound gives 
\[
    \PX(E_0) \geq 1 - 2 |R_0| |S_0| \exp(- \tilde{c}_1 {m}/{2}) \geq 1 - 2 \Big(\frac{e \, m}{\ell}\Big)^{2\ell} \exp(- \tilde{c}_1 {m}/{2}).
\]
We will next work on the event $E_0 \cap E_{R,S}$. 
Fix nonzero $r \in R$ and $s \in S$, and define the set of indices
\[
    \Omega_{r,s} := \{ j \in [m] : \langle w_j, r\rangle = 0 \; or \; \langle w_j, s \rangle = 0 \}
\]
Observe then that by the definition of $\Wpr$ and $\Omega_{r,s}$ the following holds
\[
\begin{aligned}
    \Wpr^T \Wps &= \sum_{j = 1}^m \indx_{\langle w_j, r \rangle > 0} \indx_{\langle w_j, s \rangle > 0} w_j w_j^T \\
    &= \sum_{j \in \Omega_{r,s}} \indx_{\langle w_j, r \rangle > 0} \indx_{\langle w_j, s \rangle > 0} w_j w_j^T 
    + \sum_{j \in \Omega^c_{r,s}} \indx_{\langle w_j, r \rangle > 0} \indx_{\langle w_j, s \rangle > 0} w_j w_j^T\\
     &= \sum_{j \in \Omega^c_{r,s}} \indx_{\langle w_j, r \rangle > 0} \indx_{\langle w_j, s \rangle > 0} w_j w_j^T\\
\end{aligned}
\]
On the event $E_{R,S}$, there exist therefore $r_0 \in R_0$ and $s_0 \in S_0$ such that for all $j \in \Omega_{r,s}^c$ it holds that
\[
    \operatorname{sgn}(\langle w_j, r \rangle ) = \operatorname{sgn}(\langle w_j, r_0 \rangle ) \quad \text{and} \quad  \operatorname{sgn}(\langle w_j, s \rangle ) = \operatorname{sgn}(\langle w_j, s_0 \rangle ).
\]
In particular, we can write 
\[
\begin{aligned}
   \Wpr^T \Wps  
   &= \sum_{j \in \Omega^c_{r,s}} \indx_{\langle w_j, r \rangle > 0} \indx_{\langle w_j, s \rangle > 0} w_j w_j^T\\
   &= W_{+, r_0}^T W_{+, s_0} - \sum_{j \in \Omega_{r,s}} \indx_{\langle w_j, r_0 \rangle > 0} \indx_{\langle w_j, s_0 \rangle > 0} w_j w_j^T\\
   &=: W_{+, r_0}^T W_{+, s_0} - \widetilde{W}_{+, r_0}^T \widetilde{W}_{+, s_0} 
\end{aligned}
\]
The next lemma shows that the residual $\widetilde{W}_{+, r_0}^T \widetilde{W}_{+, s_0} $ has small norm when acting on $T$. 

\begin{lemma}\label{lemma:Wtilde}
Fix $0<\eps< 1$ and $\ell < m$. Suppose that $W \in \R^{m \times n}$ has i.i.d. $\mathcal{N}(0,1/m)$ entries. 
Let $T \subset \R^{n}$ be an $\ell$-dimensional subspace and $R_0$ and $S_0$ be subsets of $\R^n$. 
Let $E_1$ be the event the following inequality holds for all set of indexes $\Omega \subset [m]$ with cardinality $|\Omega| \leq 2 \ell$:
\[
|\langle \widetilde{W}_{+, r_0}^T \widetilde{W}_{+, s_0} u, v \rangle | 
\leq \eps \|u\|_2\|v\|_2\quad \forall u,v \in T, \, r_0 \in R_0, \, s_0 \in S_0
\]
where 
\[
    \widetilde{W}_{+, r_0}^T \widetilde{W}_{+, s_0} := \sum_{j \in \Omega} \indx_{\langle w_j, r_0 \rangle > 0} \indx_{\langle w_j, s_0 \rangle > 0} w_j w_j^T.
\]  
There exists a $\delta_{\eps} > 0$ such that if $m \geq 9 \eps^{-1} \ell$ an $2 \ell \leq \delta_\eps m$, 
then $\PX (E_1) \geq 1 - 2 m \exp(- \eps m / 36)$.
\end{lemma}
We now consider the event $E := E_1 \cap E_0 \cap E_{R,S}$ where $E_1$  
is the event defined in the previous lemma. On $E$ for all $r \in R, \,s \in S$ and $u, v \in {T}$, 
\begin{align}
    |\langle {W}_{+, r}^T {W}_{+, s}u, v \rangle - \langle Q_{r,s}u, v \rangle | 
    &= \Bigl| \langle {W}_{+, r_0}^T {W}_{+, s_0} u, v \rangle - \langle \widetilde{W}_{+, r_0}^T \widetilde{W}_{+, s_0} u, v \rangle 
    - \langle Q_{r,s} u, v \rangle  \Bigr|  \notag \\
    &\leq \Bigl| \langle {W}_{+, r_0}^T {W}_{+, s_0} u, v \rangle 
    - \langle Q_{r_0,s_0} u, v \rangle \Bigr| \notag \\
    &\quad +  \Bigl| \langle Q_{r_0,s_0} u, v \rangle
    - \langle Q_{r_0,s_0} u, v \rangle \Bigr| \notag \\
    &\quad + \Bigl| \langle \widetilde{W}_{+, r_0}^T \widetilde{W}_{+, s_0} u, v \rangle \Bigr|   \notag \\
    &\leq 3 \eps \|u\|_2\|v\|_2+ \frac{60}{\pi} \eps \|u\|_2\|v\|_2+  \eps \|u\|_2\|v\|_2\notag \\ 
    &\leq 24 \eps \|u\|_2\|v\|_2\label{eq:Leps_uv},
\end{align}
where the first equality used the event $E_{R,S}$ and the definition of $\widetilde{W}_{+, r_0}^T \widetilde{W}_{+, s_0}$. 
The second inequality used instead the event $E_1 \cap E_0$ and the Lipschitz continuity of $Q_{r,s}$ (Lemma \ref{lemma:QLip}). 

In conclusion, there exist $C_\eps$ and $c_\eps$ such that if $m \geq C_\eps \ell'$ then
\[
\begin{aligned}
    \PX (E_1 \cap E_0 \cap E_{R,S}) 
    &\geq 1 - 2 m \exp(- \eps m/ 36) - 2 \Big(\frac{e \, m}{\ell}\Big)^{2 \ell}\exp(- {\tilde{c}_1 m}/{2}) 
    - 2 C_2 \exp(- c_2 \eps m) \\
    &\geq 1 - \gamma \Big(\frac{e \, m}{\ell}\Big)^{2 \ell} \exp(- {c}_\eps m)
\end{aligned}
\]
Here $C_\eps$ depends polynomially on $\eps^{-1}$ and $c_\eps$ depends polynomially on $\eps$, and $\gamma$ is positive absolute constant.
\smallskip

Notice that \eqref{eq:Leps_uv} gives a bound in terms of $24 \eps \|u\|_2 \|v\|_2$.  
To obtain a bound as in \eqref{eq:Subs_uv} simply rescale $\eps$ by $1/24$ in the discussion above, and modify $c_\eps$ and $C_\eps$ accordingly.

\smallskip
To extend \eqref{eq:Subs_uv} to the union of subspaces, we consider the subspace  $T_{i,j} = \text{span}(U_i, V_j)$ 
with dimension at most $2 \ell'$. Then use  \eqref{eq:Subs_uv} with subspaces $T_{i,j}$, $R_p$ and $S_q$, and take a union bound. 
\end{proof}

\subsection{Supplemental Results for Section \ref{sec:ProofR2WDC}}\label{appx:suppProofRRWDC}

We begin this section by providing an upper bound on the number of activation patterns of a $\relu$ layer. 
This result is used in the proof of Lemma \ref{lemma:conseqTess}.
\begin{lemma}\label{lemma:partition}
Let $S$ be an $\ell$-dimensional subspace of $\R^n$ and $m \geq \ell $. Let $W \in \R^{m \times n}$ have i.i.d $\mathcal{N}(0,1/m)$ entries. Then with probability 1, 
\[
    |\{ \diag(W s > 0) W \; | s \in  S \}|  \leq  \Big( \frac{e m}{\ell} \Big)^\ell
\]
\end{lemma}
\begin{proof}
Observe that by rotational invariance of the Gaussian distribution we may take, without loss of generality, 
$S$ to be the span of the first $\ell$ standard basis vector, i.e.\ $S = \text{span}(e_1, \dots, e_\ell)$. 
We can then also take $W \in \R^{m \times \ell}$ and $S = \R^{\ell}$.  

Let $\{ w_j \}_{j=1}^m$ be the rows of the matrix $W$.
Notice that for fixed $W$, $|\{ \diag(W s > 0) W \; | s \in  S \}|$ equals the number of binary vectors $(\indx_{\langle w_j, v\rangle > 0})_{j \in [n]}$ for $v \in \mathcal{S}^{\ell - 1}$.  Each $(\indx_{\langle w_j, v\rangle > 0})_{j \in [n]}$ uniquely identifies a region of the partitioning of $\R^l$ induced by the set of hyperplanes $\mathcal{H} := \{ x: \, \langle w_j,  x \rangle = 0 \}$. From the theory of hyperplane arrangements \cite{matousek2013lectures} we know that $m \geq \ell$ hyperplanes in $\R^\ell$ partition the space in at most $\sum_{j = 0}^\ell \binom{m}{j}$. Thus, with probability 1 we have 
	\[
	\begin{aligned}
		|\{ \diag(W s > 0) W \; | s \in  S \}|  
		&\leq \sum_{j=0}^\ell \binom{m}{j}\\
		&\leq \sum_{j=0}^\ell \frac{m^j}{j!} \leq \sum_{j=0}^\ell \frac{\ell^j}{j!} \Big( \frac{m}{\ell} \Big)^j 
		\leq \Big( \frac{m}{\ell} \Big)^\ell \sum_{j=0}^\infty \frac{\ell^j}{j!} 
		= \Big( \frac{e m}{\ell} \Big)^\ell
	\end{aligned}
	\]
\end{proof}

Next we prove Lemma \ref{lemma:Wtilde}, providing an upper bound for the random matrix $\widetilde{W}^T \widetilde{W}$ 
when acting on low-dimensional subspaces. 

\begin{proof}[Proof of Lemma \ref{lemma:Wtilde}]
	Notice that for any $\Omega \subset [m]$, $u, v \in T$, $r_0 \in R_0$ and $s_0 \in S_0$, it holds that
	\[
	\begin{aligned}
		|\langle \widetilde{W}_{+, r_0}^T\widetilde{W}_{+, s_0} u, v  \rangle| &= 
		| \langle \text{diag}(W_\Omega r_0 > 0) \odot \text{diag}(W_\Omega s_0 > 0) W_\Omega u, W_\Omega v \rangle | \\
		&\leq \|\text{diag}(W_\Omega r_0 > 0) \odot \text{diag}(W_\Omega s_0 > 0)\| \| W_\Omega v\| \| W_\Omega u\|\\
		&\leq \| W_\Omega v\| \| W_\Omega u\|.
	\end{aligned}
	\] 
	Therefore, it is sufficient to show that 
	\[
		\| W_\Omega u \| \leq \sqrt{\eps} \|u\| \qquad \forall u \in T \: \; \forall \Omega \subset [m] \;\; \text{satisfying} \;\; |\Omega| \leq 2 \ell \leq \delta_\eps m.
	\]
	The rest of the proof follows, \textit{mutatis mutandis}, as in Lemma 26 of \cite{leong2020}. 
\end{proof} 

We will next show that $Q_{x,y}$ is a Lipschitz function of its arguments. 

\begin{lemma}\label{lemma:QLip}
    Fix $0 < \eps < 1$ and $x, \tilde{x}, y, \tilde{y} \in \mathcal{S}^{n-1}$. If $\| \tilde{x} - x\| \leq \eps$ 
    and $\| \tilde{y}  - y \|\leq \eps$, then
    \[
        \| Q_{\tilde{x},\tilde{y}} - Q_{\tilde{x}, \tilde{y}} \| \leq \Big(\frac{2}{\pi} + 2 \sqrt{79}\Big) \eps
    \]
\end{lemma}
\begin{proof}
    Recall the following facts:
    \begin{align}
        \|x - y \| &\geq 2 \sin( \angle (x,y)/2 ),  &\forall x,y \in \mathcal{S}^{n-1} \label{eq:xySinxy}\\
        |\angle(x_1, x_2)| &\geq |\angle(x_1,y) - \angle(x_2,y)|,  &\forall x_1, x_2, y \in \mathcal{S}^{n-1} \label{eq:triangleAngle}\\
        \sin(\theta/2) &\geq \theta/4,  &\forall \theta \in [0, \pi] \label{eq:sintheta}
    \end{align}
    Let $\theta_{\tilde{x}, x} = \angle(\tilde{x}, x)$ and $\theta_{\tilde{y}, y} = \angle(\tilde{y}, y)$, then
    \[
        \| Q_{x,y} - Q_{\tilde{x}, \tilde{y}}\| \leq \frac{|\theta_{ x, y} - \theta_{\tilde{x}, \tilde{y}}|}{2 \pi} + \Big\| \frac{\sin \theta_{x,y}}{2\pi} M_{x \leftrightarrow y} - \frac{\sin \theta_{\tilde{x},\tilde{y}}}{2\pi} M_{\tilde{x} \leftrightarrow \tilde{y}} \Big\|.
    \]
    By \eqref{eq:triangleAngle} it holds that
    \[
    	|\theta_{ x, y} - \theta_{\tilde{x}, \tilde{y}}| \leq |\theta_{ x, y} - \theta_{ \tilde{x}, y}| + |\theta_{\tilde{x}, {y}} - \theta_{\tilde{x}, \tilde{y}}| \leq |\theta_{\tilde{x}, x}| + |\theta_{\tilde{y}, y}|,
    \]
    while from \eqref{eq:xySinxy} and \eqref{eq:sintheta} it follows that
    \begin{align*}
        |\theta_{\tilde{x}, x}| &\leq 4 \sin(\theta_{\tilde{x}, x}/2) \leq 2 \eps,\\
        |\theta_{\tilde{y}, y}| &\leq 4 \sin(\theta_{\tilde{y}, y}/2) \leq 2 \eps.
    \end{align*}
    Thus $|\theta_{ x, y} - \theta_{\tilde{x}, \tilde{y}}| \leq 4 \eps$.
    Lemma B.3 in \citep{daskalakis20} then proves that
    \[
        \Big\| \frac{\sin \theta_{x,y}}{2\pi} M_{x \leftrightarrow y} - \frac{\sin \theta_{\tilde{x},\tilde{y}}}{2\pi} M_{\tilde{x} \leftrightarrow \tilde{y}} \Big\| \leq 2 \sqrt{79} \eps,
    \]
    which concludes the proof.
\end{proof}

\section{Proof of Lemma \ref{lemma:noise1}}\label{appx:noise}
In this section we prove Lemma \ref{lemma:noise1} which is used to bound the perturbation of the objective function $\fCS$ 
and its gradient due to the presence of the noise term $\eta$. 

\begin{proof}[Proof of Lemma \ref{lemma:noise1}]
Fix $x, z \in \mathcal{S}^{k-1}$ and notice that by the properties of the Gaussian distribution, for $t \geq 0$ it holds that
   \[
    \PX_{A}\big[ \langle z, \Lambda_x^T A^T \eta \rangle \geq \frac{\|\Lambda_x z\|}{\sqrt{m}} \|\eta\| t \big] = \PX_{y \sim \mathcal{N}(0,1)}\Big[ \frac{\|\Lambda_x z\|}{\sqrt{m}} \|\eta\| y \geq  \frac{\|\Lambda_x z\|}{\sqrt{m}} \|\eta\| t \Big] \leq e^{- \frac{t^2}{2}}.
\]
If $z = x$ use \eqref{eq:boundGx}, while if $z \neq x$ and $G$ differentiable at $x$ use  \eqref{eq:Mdnorm}, to obtain that 
\[
    \PX_{A}\Big[ \langle z, \Lambda_x^T A^T \eta \rangle \geq \sqrt{\frac{13}{12}} \frac{\|\eta\|}{2^{d/2}} \frac{t}{\sqrt{m}} \Big] \leq e^{- \frac{t^2}{2}}
\]
Let $\mathcal{N}_{1/2}$ be a $\frac{1}{2}$-net over $\mathcal{S}^{k-1}$ such that $|\mathcal{N}_{1/2}|\leq 5^{k}$ 
(see for example \citep{vershynin2018high}). 
Recall that by Lemma \ref{lemma:joshiRange} the number of different matrices $\Lambda_x$ is bounded by $\Psi_d$.  Thus, a union bound gives 
\[
\begin{aligned}	
\PX \Big[ \langle z, \Lambda_x^T A^T \eta \rangle \geq \sqrt{\frac{13}{12}} \frac{\|\eta\|}{2^{d/2}} \frac{t}{\sqrt{m}}, \; \; \forall x, z \in \mathcal{S}^{k-1}  \Big] 
    &\leq 
    |\mathcal{N}_{\frac{1}{2}}| \, \Psi_d \, \PX \Big[ \langle z, \Lambda_x^T A^T \eta \rangle \geq \sqrt{\frac{13}{12}} \frac{\|\eta\|}{2^{d/2}} \frac{t}{\sqrt{m}} \Big]  \\
    &\leq \exp(- \frac{t}{2}^2 + \log 5 + \log \Psi_d)
\end{aligned}
\]
Choosing $t = 2 \sqrt{k \log(5 \prod_{i=1}^{d} \frac{e \,n_i}{k})}$ we obtain the theses.


\end{proof}

\section{Extensions}\label{appx:extensions}
\subsection{Compressive Phase Retrieval with a Generative Prior}\label{appx:PR}

Consider a generative network $G: \R^{k} \to \R^{n}$ as in \eqref{eq:Gx}. 
The compressive phase retrieval problem with a generative network prior can be formulated as follows. 
\begin{tcolorbox}
\begin{center}
\textbf{COMPRESSIVE PHASE RETRIEVAL WITH A DEEP GENERATIVE PRIOR}
\end{center}
 \begin{tabular}{rl}
 \textbf{Let}: & $G: \R^k \to \R^n$ generative network, $A \in \R^{m\times n}$ measurement matrix. \\
 \textbf{Let}: & $\ystar = G(\xstar)$ for some unknown $\xstar \in \R^k$.\\
 {} & {} \\
 \textbf{Given}: & $G$ and $A$. \\
  \textbf{Given}: & {Measurements} $b = |A \ystar| + \eta \in \R^{m}$ {with} $m \ll n$ and  $\eta \in \R^{m}$ {noise}.\\
   {} & {} \\
  \textbf{Estimate}: & $\ystar$.
  \end{tabular}
\end{tcolorbox}

To estimate $\ystar$, \cite{oscar2018phase} proposes to find the latent code $\hat{x}$  that minimizes the reconstruction error

\begin{align}
    \tilde{x} &= \arg \min_{x \in \R^{x}} \fPR(x) := \frac{1}{2} \| b - |A G(x)| \|_2^2,  \label{eq:oscar} \\
    \ystar &\approx G(\tilde{x}). \nonumber
\end{align}
In \cite{leong2020} it is shown that Algorithm \ref{algo:subGrad} with inputs $\fPR$, 
small enough step size and arbitrary initial condition estimates $\ystar$ up to the noise level in polynomial time, 
provided that the number of phaseless measurements is up-to $\log$-factors $m \geq k \cdot \text{poly}(d)$ 
and the generative network is logarithmically expansive. The proof uses the WDC and an isometry condition akin to the RRIC. 
As before, the RWDC can be replaced by the RRWDC and obtain the same convergence guarantees. 
Moreover, as in the case of compressed sensing, the logarithmic factor in the number of measurements can be improved using Lemma \ref{lemma:joshiRange}.

\subsection{Denoising with a Generative Prior}\label{appx:Den}

Consider a generative network $G: \R^{k} \to \R^{n}$ as in \eqref{eq:Gx}. 
The denoising problem with a generative network prior can be formulated as follows. 
\begin{tcolorbox}
\begin{center}
\textbf{DENOISING WITH A DEEP GENERATIVE PRIOR}
\end{center}
 \begin{tabular}{rl}
 \textbf{Let}: & $G: \R^k \to \R^n$ generative network. \\
 \textbf{Let}: & $\ystar = G(\xstar)$ for some unknown $\xstar \in \R^k$.\\
 {} & {} \\
 \textbf{Given}: & $G$. \\
  \textbf{Given}: & {Noisy signal} $b =  \ystar + \eta \in \R^{m}$ {with} $\eta \sim \mathcal{N}(0, \sigma^2 \Id_n)$ {noise}.\\
   {} & {} \\
  \textbf{Estimate}: & $\ystar$.
  \end{tabular}
\end{tcolorbox}
To estimate $\ystar$, \cite{heckel2021rate} proposes to find the latent code $\hat{x}$  that minimizes the reconstruction error
\begin{align}
    \tilde{x} &= \arg \min_{x \in \R^{x}} \fDen(x) := \frac{1}{2} \| b -  G(x) \|_2^2,  \label{eq:heckel} \\
    \ystar &\approx G(\tilde{x}). \nonumber
\end{align}
In \cite{heckel2021rate} recovery guarantees based on this minimization problem  
are given for an expansive generative network $G$.  Specifically, it is shown that Algorithm \ref{algo:subGrad} with input $\fDen$, 
small enough step size $\alpha$ and arbitrary initial point $x_0$, reconstructs the signal $\ystar$ up to an $O(k/n)$ error. 
The random network $G$ is assumed to be logarithmically expansive in order to satisfy the WDC with high-probability, 
but inspecting the proof it can be seen that the \RRWDC is enough. Using Lemma \ref{prop:RandRRWDC} 
we can extend the result of \cite{heckel2021rate} to the case of a generative network 
satisfying Assumptions \ref{AssumptionB}. 
\subsection{Spiked Matrix Recovery with a Generative Prior}\label{appx:Spike}

Consider a generative network $G: \R^{k} \to \R^{n}$ as in \eqref{eq:Gx}. 
The spiked Wishart matrix recovery with a generative prior is formulated as follows. 
\begin{tcolorbox}
\begin{center}
\textbf{SPIKED WISHART MATRIX RECOVERY \\ WITH A DEEP GENERATIVE PRIOR}
\end{center}
 \begin{tabular}{rl}
 \textbf{Let}: & $G: \R^k \to \R^n$ generative network. \\
 \textbf{Let}: & $\ystar = G(\xstar)$ for some unknown $\xstar \in \R^k$.\\
 {} & {} \\
 \textbf{Given}: & $G$. \\
  \textbf{Given}: & {Noisy matrix} $B = u \, \ystar^T + \sigma \mathcal{Z} \in \R^{N \times n}$, {with} $u \sim \mathcal{N}(0, I_N)$\\
  {} & and $\mathcal{Z}$ with i.i.d.\ $\mathcal{N}(0,1)$ entries.\\
   {} & {} \\
  \textbf{Estimate}: & $\ystar$.
  \end{tabular}
\end{tcolorbox}

Similarly, the spiked Wigner matrix recovery with a generative prior is formulated as follows. 
\begin{tcolorbox}
\begin{center}
\textbf{SPIKED WIGNER MATRIX RECOVERY \\ WITH A DEEP GENERATIVE PRIOR}
\end{center}
 \begin{tabular}{rl}
 \textbf{Let}: & $G: \R^k \to \R^n$ generative network. \\
 \textbf{Let}: & $\ystar = G(\xstar)$ for some unknown $\xstar \in \R^k$.\\
 {} & {} \\
 \textbf{Given}: & $G$. \\
  \textbf{Given}: & {Noisy matrix} $B = \ystar \, \ystar^T + \sigma \mathcal{H} \in \R^{n \times n}$, \\
  {} & {with} $\mathcal{H}$ from a Gaussian Orthogonal Ensemble.\\
   {} & {} \\
  \textbf{Estimate}: & $\ystar$.
  \end{tabular}
\end{tcolorbox}

To estimate $\ystar$, \cite{cocola2020nonasymptotic} proposes to find the latent code $\hat{x}$  that minimizes the reconstruction error
\begin{align*}
    \tilde{x} &= \arg \min_{x \in \R^{x}} \fSpike(x) := \frac{1}{2} \| M -  G(x)G(x)^T \|_F^2, \\
    \ystar &\approx G(\tilde{x}), \nonumber
\end{align*}
where
\begin{itemize}
	\item in the spiked Wishart model $M = B^T \,B/N - \sigma^2 I_n$;
	\item in the spiked Wigner model $M = B$. 
\end{itemize}

As shown in \cite{cocola2021no}  Algorithm \ref{algo:subGrad} with inputs $\fSpike$, appropriate $\alpha$ and arbitrary initial point $x_0$, 
estimates in polynomial time the signal $\ystar$ with rate-optimal dependence on the noise level or sample complexity. 
In particular, this shows that the absence of a computational-statistical gap in spiked matrix recovery with an expansive (random) 
generative network prior. The proof uses the fact that for $G$ satisfying the WDC the bounds in  
Proposition \ref{prop:concGandTheta}, \ref{prop:grad_conctr}, \ref{prop:Lip} and \ref{prop:convx} hold. 
Since these bounds hold under the weaker RRWDC we can directly extend the results in \cite{cocola2021no} 
to non-expansive generative networks $G$ satisfying Assumptions \ref{AssumptionB}. 

\section{An example of a contractive generative network}\label{appx:example}

In this section we give an example of a generative network as in \eqref{eq:Gx} 
satisfying the conditions \eqref{eq:expansivity} and \eqref{eq:addAssmp}, and with contractive layers. 
\smallskip

Let $d \geq 2$ and $\bar{C}_\eps := \max(\widetilde{C}_\eps, {16 c_\eps^{-1}}/{\log(2)})$. 
Then consider a $d$-layer generative network $G$ such that for $i \in [d]$
\[
	n_i := \bar{C}_\eps \cdot k \cdot d (2 d - i) \cdot \alpha,  
\]
where $\alpha \cdot \bar{C}_\eps \in \mathbb{N}$ and 
\begin{equation}\label{eq:maxalpha}
	\alpha \geq \max\bigg\{ \frac{2 \log\big(\bar{C}_\eps \cdot k \big)}{d^2}, \log\big( e^2 \bar{C}_\eps \big) \bigg\}.
\end{equation}
We now demonstrate that $n_i$ satisfies \eqref{eq:expansivity}. Notice that 
\[
\begin{aligned}
	\log\Big( \prod_{j=1}^{i-1}  \frac{e n_j}{k}\Big) &= \sum_{j=1} ^{i-1} \log\big( \alpha \cdot \bar{C}_\eps \cdot d (2 d - j) \cdot e \big) \\
	&\leq (d-1)  \log\big( \alpha \cdot \bar{C}_\eps \cdot  2d^2 \cdot e \big) \\
	&= (d-1) \big[ \log(e\bar{C}_\eps) + 2 \log(d) \big] + (d-1) \log( 2\alpha)\\
	&\leq (d-1)d  \big[ \log(e\bar{C}_\eps) + 1 \big] + (d-1) \alpha \\
	&\leq (d-1)d  \alpha + (d-1) \alpha\\
	&= (d^2 -1)\alpha
\end{aligned}
\]
where in the second inequality we have used $2 \log(x) \leq x$ and $\log(2 x) \leq x$ for $x > 0$ 
and in the third \eqref{eq:maxalpha}. Next since $d (2 d - i) \geq (d^2 -1)$ 
for every $i \in [d]$, $n_i$ satisfies \eqref{eq:expansivity} for every $i \in [d]$. 

We now show that $n_i$ satisfies \eqref{eq:addAssmp}. We have
\[
\begin{aligned}
	\log(n_i) &= \log(\bar{C}_\eps \cdot k \cdot d (2 d - i) \cdot \alpha)\\
	&= \log( d (2 d - i)  \alpha) + \log(\bar{C}_\eps  k) \\
	&\leq \frac{ d (2 d - i)  \alpha}{2} + \log(\bar{C}_\eps  k) \\
	&\leq \frac{ d (2 d - i)  \alpha}{2} + \frac{ d^2 \alpha}{2} \\
	&\leq d (2 d - i) \alpha,
\end{aligned}
\]
where in the first inequality we have used $2 \log(x) \leq x$ for $x>0$, 
in the second inequality \eqref{eq:maxalpha} and in the third $2 d^2 \leq d (2 d - i) $ for every $i \in [d]$. We therefore have 
\[
	\log(n_i) \cdot  \frac{16 \cdot k \cdot c_\eps^{-1} }{\log(2)}
 \leq d (2 d - i) \cdot \alpha \cdot \frac{16 \cdot k \cdot c_\eps^{-1} }{\log(2)} \leq d (2 d - i) \cdot \alpha \cdot  \bar{C}_\eps \cdot k = n_i.
\]

\end{document}